\documentclass[12pt,draftclsnofoot,twoside,onecolumn,letter]{IEEEtran}

\usepackage{graphicx}
\usepackage{epstopdf}
\DeclareGraphicsExtensions{.eps,.pdf} 
\graphicspath{ {figures/} }
\usepackage{cite}
\usepackage{subfigure}
\usepackage{amssymb,amsmath}
\usepackage{algorithm}
\usepackage{algorithmic}
\usepackage{color}
\usepackage{multirow}
\usepackage{multicol}
\usepackage{cases}
\usepackage{setspace}

\usepackage{epsfig}
\usepackage[T1]{fontenc}
\usepackage{amsmath}
\usepackage{amsfonts}
\usepackage{amsthm,amssymb}

\makeatletter
\newcommand\restartchapters{\par
  \setcounter{chapter}{0}%
  \setcounter{section}{0}%
  \gdef\@chapapp{\chaptername}%
  \gdef\thechapter{\@arabic\c@chapter}}
\makeatother

\newtheorem{theorem}{\bf {Theorem}}

\newtheorem{remark}{{\bf{Remark}}}

\newtheorem{lemma}{\bf {Lemma}}

\newtheorem{proposition}{{\it Proposition}}

\renewcommand{\algorithmicrequire}{\textbf{Input:}}

\newcommand*{\hili}{\color{black}}

\usepackage{capt-of}
\usepackage{dblfloatfix}
\usepackage{varwidth}
\newcommand{\PS}{{P}_{\mathrm{s}}}

\newcommand{\PSb}{{P}_{\rm{s_b}}}
\newcommand{\PH}{{P}_{\rm{H}}}

\newcommand{\hPSb}{\textbf{h}_{{\rm{p}}\rm{s_b}}^{H}}

\newcommand{\hSbSb}{h_{{\rm{s_b}}\rm{s_b}}}
\newcommand{\hSbRi}{h_{{\rm{s_b}{\rm{r}}}_{i}}}

\newcommand{\gSbP}{|h_{\rm{s_b} {\rm{p}}}|^2}
\newcommand{\gSbRm}{|h_{\rm{s_b} {\rm{r}_\textit{m}}}|^2}

\newcommand{\PT}{\rm{p}}
\newcommand{\PR}{\rm{p}}

\newcommand{\Sbb}{\rm{s_b}}
\newcommand{\Sb}{{\rm{ST}}_{\rm{b}}}
\newcommand{\Rm}{{\rm{r}}_m}
\newcommand{\Ri}{{\rm{r}}_i}

\newcommand{\w}{\textbf{w}_{\rm{p}}}

\newcommand{\subparagraph}{}

\usepackage{titlesec}

\begin{document}
\bstctlcite{IEEEexample:BSTcontrol}

\title{{\hili Full-Duplex Non-Orthogonal Multiple Access Cooperative Overlay Spectrum-Sharing Networks with SWIPT}}
\author{
\IEEEauthorblockN{Quang Nhat Le, Animesh Yadav, Nam-Phong Nguyen, Octavia~A.~Dobre, and Ruiqin Zhao \vspace{-30pt}}
\\
\thanks{Part of this work was presented at the IEEE Global Communications Conference, Waikoloa, HI, Dec. 2019 \cite{in}.}
\thanks{Q. N. Le and O. A. Dobre are with the Dept. of Electrical and Computer Engineering, Memorial University, St. John’s, NL A1B 3X9, Canada (e-mail: \{qnle, odobre\}@mun.ca).}
\thanks{A. Yadav is with the Dept. of Electrical Engineering and Computer Science, Syracuse University, Syracuse, NY 13244, USA (e-mail: ayadav04@syr.edu).}
\thanks{N-P. Nguyen is with the School of Electronics and Telecommunications, Hanoi University of Science and Technology, Hanoi, Vietnam (e-mail: phong.nguyennam@hust.edu.vn).}
\thanks{R. Zhao is with the School of Marine Science and Technology, Northwestern Polytechnical University, Xi’an 710072, China (e-mail: rqzhao@nwpu.edu.cn).}
	}

\maketitle
\begin{abstract}
This paper proposes a novel non-orthogonal multiple access (NOMA) assisted cooperative spectrum-sharing network, in which one of the full-duplex (FD) secondary transmitters (STs) is chosen among many for forwarding the primary transmitter's and its own information to primary receiver and secondary receivers, respectively, using NOMA technique. To stimulate the ST to conduct cooperative transmission and sustain its operations, the simultaneous wireless information and power transfer (SWIPT) technique is utilized by the ST to harvest the primary signal's energy. In order to evaluate the proposed system's performance, the outage probability and system throughput for the primary and secondary networks are derived in tight closed-form approximations. Further, the sum rate optimization problem is formulated for the proposed cooperative network and a rapid convergent iterative algorithm is proposed to obtain the optimized power allocation coefficients. Numerical results show that FD, SWIPT, and NOMA techniques greatly boost the performance of cooperative spectrum-sharing network in terms of outage probability, system throughput, and sum rate compared to that of half-duplex NOMA and the conventional orthogonal multiple access-time division multiple access networks.    
\end{abstract}


\section{Introduction}

Non-orthogonal multiple access (NOMA) has been recognized as a potential spectral-efficiency improving technique for the fifth-generation (5G) and beyond wireless networks \cite{13}-[5]. Its underlying principle enables multiple users to concurrently access and transmit their signals in the same spectrum resource block (i.e., time/frequency/code domain) by using different signal signatures for the case of code-domain NOMA (CD-NOMA) or power levels for the case of power-domain NOMA (PD-NOMA). Message passing algorithm and successive interference cancellation (SIC) are used to separate the superimposed signals for CD-NOMA and PD-NOMA, respectively. Since NOMA\footnote{In this work, we use power-domain NOMA and henceforth, we will be referring to it as NOMA.} can enhance the spectral efficiency, user fairness, and realize massive connectivity compared to the conventional orthogonal multiple access (OMA) scheme, it can greatly enhance the performance of wireless networks \cite{13,15}. In \cite{15}, a cooperative relaying system in a device-to-device-NOMA was proposed and its performance was evaluated in terms of scaled system capacity. The proposed system achieved much higher ergodic capacity compared to the conventional OMA system.   

For 5G and beyond networks, in addition to the prerequisite of high spectral efficiency, providing energy-efficient communications is also an important goal. Further, since wireless users are battery-operated, it is more likely that performing cooperative transmisison task will lead to rapid energy exhaustion of their batteries. In order to help mobile users maintain their operations, simultaneous wireless information and power transfer (SWIPT) \cite{19,20} has come out as an efficient way to provide energy and extend the lifetime of energy-constrained wireless devices. In \cite{19}, the authors proposed a SWIPT-based NOMA network in which near users being close to the source acted as energy scavenging relays to help the source forward the data to far users. Closed-form expressions for the outage probability and system throughput were obtained to assess the system performance. The authors in \cite{20} developed the system model where the relay user scavenged energy from the base station (BS)'s NOMA signal and used harvested power to forward the information to the destination. The outage probability was used as a performance metric to evaluate their proposed system's performance. However, the work in \cite{15}-\cite{20} considered half-duplex (HD) relaying mode, where relay nodes cannot simultaneously receive and transmit information in the same frequency band. 

With the advancement in antenna and signal processing technologies, full-duplex (FD) communication mode has attracted much research interest due to its ability in doubling the spectral efficiency by allowing users to simultaneously receive and transmit data in the same frequency band \cite{24}-[11]. In \cite{21}, the authors proposed the user-assisted cooperative NOMA system in which the strong user operating in FD or HD mode forwarded the information message to the weak user. By analysing the system outage probability and ergodic sum rate, the authors concluded that FD NOMA was superior to HD NOMA in the low signal-to-noise ratio (SNR) region. In \cite{4}, the authors considered SWIPT in a cooperative FD NOMA system where the near user relayed the message to the far user by harvesting the radio frequency (RF) energy. The result showed that the effect of self-interference (SI) signal in FD communications was mitigated thanks to SWIPT since additional gain could be achieved. Besides, cognitive radio (CR) is another potential technology that significantly improves the spectrum efficiency by allowing unauthorized secondary users (SUs) to access authorized primary users's (PUs) spectrum. Hence, integrating NOMA to CR networks has a highly conceivable possibility to provide an efficient spectrum use, so that the requirements of 5G and beyond wireless networks, i.e., high spectral efficiency, low latency, and massive connectivity, can be readily achieved \cite{14,18}. Considering multiple-input multiple-output (MIMO) CR-NOMA system, the authors in \cite{3} proposed a novel joint antenna selection algorithm to further enhance the system performance. In \cite{23}, a cooperative multicast for CR-NOMA scheme was developed to improve the outage probability of PU. In order to maximize the harvested energy of SUs and based on a practical non-linear energy harvesting (EH) model, the authors in \cite{26} proposed an optimal resource allocation strategy in a SWIPT-CR-NOMA network, where SUs shared the licensed spectrum with PUs under the condition that the interference caused by SUs was acceptable. In \cite{5}, the cooperative NOMA relay-supported CR network was investigated where SU helped PU by acting as relays and exploited NOMA technique to transmit the PU's and its own messages together to the destination using the licensed spectrum band. Under the same system model as in \cite{5}, \cite{6} focused on user scheduling schemes to improve the outage performance for both primary and secondary systems.           

Unlike the existing works \cite{3}-\cite{26},\cite{7} which considered NOMA underlay CR networks, we propose a novel FD NOMA assisted cooperative overlay spectrum-sharing system with SWIPT that encourages the cooperation between primary and secondary networks. Given that FD, SWIPT, and NOMA are potential spectral and energy efficiency improving technologies for the beyond fifth-generation wireless networks, applying such technologies will not only boost the system performance, but also showcase that they can be operated in tandem in future networks. Although there exists a few works \cite{5,6} which also considered the application of NOMA to overlay CR network, the major differences of our work compared with the existing ones are as follows. Firstly, besides spectrum sharing to secondary transmitters (STs), in order to stimulate STs to perform cooperative relaying functions and sustain their operations in terms of energy, SWIPT is applied in our proposed system where STs can harvest RF energy from the primary transmitter (PT)'s signal and use it for their relaying operations. Secondly, STs operate in FD mode, which allows them to concurrently receive and transmit signals in the same transmission time. As a result, significant increase in the spectral and energy efficiencies of the system is readily achieved. Thirdly, the effect of STs scheduling on the system performance is considered, where the best ST is selected to maximize the harvested power and improve the reception quality of primary and secondary networks. The main contributions of our paper are summarized as follows:

\begin{itemize}
\item We propose an opportunistic ST selection method to choose one best FD ST, which has the best channel connection to the PT, among multiple ones. Besides, to the best of our knowledge, this is the first work considering the application of FD, SWIPT, and NOMA in cooperative spectrum-sharing networks.
\item We characterize the performance of primary and secondary networks in terms of outage probability and system throughput over Rayleigh fading channels. To this end, the tight closed-form approximation expressions are obtained for outage probability and system throughput of both primary and secondary networks. Through numerical results, it is shown that the application of FD and SWIPT techniques significantly improve the proposed system performance compared to HD and the conventional OMA-time division multiple access (OMA-TDMA) schemes. 
\item We propose an efficient algorithm which solves the power allocation problem to maximize the sum rate of primary and secondary networks. The obtained results show the superior performance of FD compared to its HD counterpart.    
\end{itemize}

The remainder of this paper is organized as follows. Section \ref{sec:sys} describes the system model. In Section \ref{sec:per}, the outage probability and system throughput for both primary and secondary networks are successfully derived. In Section \ref{sec:sum}, a novel power allocation optimization algorithm is proposed to maximize the sum rate of primary and secondary networks. Numerical results and discussions are shown in Section \ref{sec}. Finally, the paper is concluded in Section \ref{sec:con}.

\textit{Notations:} Bold lowercase letters denote vectors and lowercase characters stand for scalars. $||\cdot||$, $(\cdot)^H$ and $|\cdot|$ correspond to the Euclidean norm, the Hermitian operator, and the absolute value, respectively. $\mathbb{E}[\cdot]$ represents the expectation operation, $\Pr(\cdot)$ denotes probability, and $\mathbb{C}$ is the set of complex-valued numbers.

\section{System Model}\label{sec:sys}
\subsection{System Description}

\begin{figure}
	\begin{center}
		\epsfxsize=8cm \epsfbox{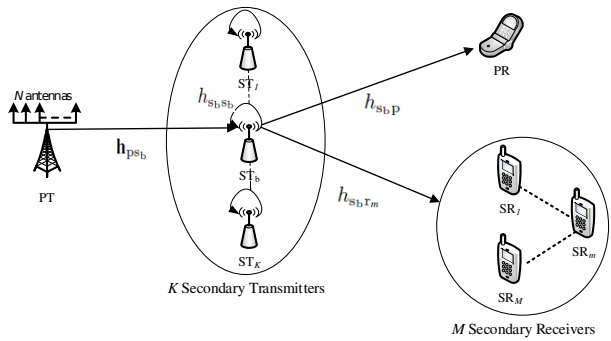} \caption{Illustration of FD NOMA-supported cooperative overlay CR network.}\label{fig:sys1}
	\end{center}
\end{figure}

As illustrated in Fig. 1, we consider a cooperative overlay spectrum-sharing system consisting of one pair of primary transceivers denoted by PT and PR, \textit{K} STs denoted by ST$_{k}$, $k=1,2,\ldots,\textit{K}$, and \textit{M} secondary receivers (SRs) denoted by SR$_{m}$, $m=1,2,\ldots,\textit{M}$. All STs operate in FD mode and others operate in HD mode. The PT is equipped with $N$ antennas, $n=1,2,\ldots,N$, while the PR and \textit{M} SRs are equipped with one transmitting/receiving antenna. Each ST has two antennas, one for receiving and the other one for transmitting \cite{4,2}. The direct link between the PT and the PR does not exist since the PR is far away from the PT \cite{6,11,8}. Hence, in order to establish communication between PT and PR, we consider an overlay spectrum-sharing scenario where the PT allows STs to access its spectrum resources as a reward for improving the primary reception by cooperative relaying. A best ST (denoted henceforth by ST$_{\rm{b}}$) is selected among \textit{K} STs to  concurrently transmit the primary information together with its own data to the PR and \textit{M} SRs by employing the NOMA technique. The aim of selecting the ST$_{\rm{b}}$ is to obtain the best primary and secondary outage performances. Besides, in order to further encourage the ST$_{\rm{b}}$ to conduct the cooperative relaying function, SWIPT is utilized in our model where the ST$_{\rm{b}}$ can harvest energy from the PT signal and use it for relaying purpose \cite{11}.

The transmission time is partitioned into equal transmission time slots of duration $T$. All wireless channels undergo Rayleigh block fading with coherence time of $T$. The channel coefficients are independent and identically distributed (i.i.d.) from one slot to the next. The channel coefficients from the ST$_{\rm{b}}$ to the PR and the SR$_{m}$ are denoted by complex scalars $h_{{\rm{s}}_b\rm{p}}$ and $h_{{\rm{s}}_b{\rm{r}}_m}$, respectively. Under Rayleigh fading model, the channel gains $|h_{{\rm{s}}_b\rm{p}}|^2$ and $|h_{{\rm{s}}_b{\rm{r}}_m}|^2$ are exponential random variables (RVs) with mean $\mathbb{E}[|h_{{\rm{s}}_b\rm{p}}|^2]=\lambda _{\rm{sp}}$ and $\mathbb{E}[|h_{{\rm{s}}_b{\rm{r}}_m}|^2]=\lambda _{\rm{sr}}$, respectively. The channel vector from the PT to the ST$_{k}$ is denoted by $\textbf{h}_{{\rm{p}\rm{s}}_k} \in \mathbb{C}^{N \times 1}$, where each element follows the Rayleigh distribution. Therefore, the channel gain $||\textbf{h}_{{{\rm{p}}\rm{s}}_k}||^{2}$ follows the Gamma distribution with parameter ($N$, $\lambda _{\rm{ps}}$), where $\lambda _{\rm{ps}}$ denotes the mean. Transmit beamforming is used for PT-ST$_{\rm{b}}$ link to enhance the reception quality of the ST$_{\rm{b}}$. The noise at each receiver is modeled as additive white Gaussian noise (AWGN) with zero mean and variance $\sigma_n$ \cite{5,6}. Further, all nodes are assumed to have perfect channel state information (CSI) to other nodes \cite{5,6,11,8}.    

\subsection{Signal Model}

At the start of the transmission time slot $T$, the ST$_{\rm{b}}$ is selected by the PT according to the following selection criterion:
\begin{equation}
\Sb = \mathop {\arg \max }\limits_{k=1,2,\ldots,K} ||\textbf{h}_{{\rm{p}}{\rm{s}}_{k}}||^2. 
\end{equation}  

Practically, the PT can obtain the CSIs from the $K$ STs by first sending the pilot signals to STs, then STs will estimate and feedback their CSIs to the PT. Thereafter, the PT will select a ST (ST$_{\rm{b}}$) which has the best connection to it. 

After choosing the ST$_{\rm{b}}$, the PT beamforms its signal to the ST$_{\rm{b}}$. The received RF signal sent by the PT at the ST$_{\rm{b}}$ is given by
\begin{equation}\label{eq:PTSn1}
y_{\PT,\Sbb} = \hPSb \w \sqrt{\PS} x_{0} + \hSbSb \sqrt{\PSb} x_s + n_{\Sbb},
\end{equation}
where $\w = \textbf{h}_{{{\rm{p}}\rm{s}}_{\rm{b}}} / ||\textbf{h}_{{{\rm{p}}\rm{s}}_{\rm{b}}}||$ denotes the transmit beamforming vector, $x_{0}$ is the PT's information signal with $\mathbb{E}[|x_{0}|^2]=1$, $x_s$ is the composite transmit signal\footnote{We assume that the FD STs decode the received signal without any delay, and hence, the decoded PT's signal $x_0$ is included in $x_s$.} of the ST$_{\rm{b}}$ including signals of the PR and $M$ SRs with $\mathbb{E}[|x_{s}|^2]=1$. $\PSb$ denotes the transmit power of the ST$_{\rm{b}}$, $\hSbSb$ represents the SI channel, and $n_{\Sbb}$ is the AWGN antenna noise at the ST$_{\rm{b}}$. 

\begin{figure}
	\begin{center}
		\epsfxsize=7cm \epsfbox{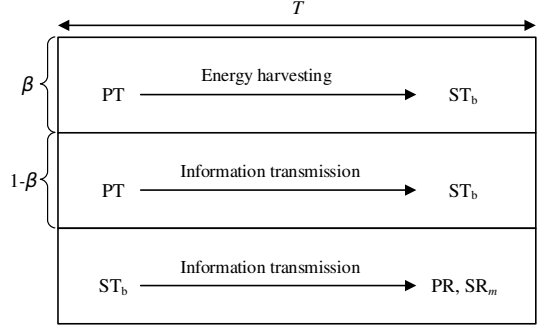} \caption{Block time $T$ for the FD ST$_{\rm{b}}$.}\label{fig:sys2}
	\end{center}
\end{figure}

Fig. 2 describes the transmission time slot $T$ at the ST$_{\rm{b}}$, where the fraction $\beta$ ($0 < \beta < 1$) of the received RF signal power from the PT is used for EH and the remaining (1-$\beta$) fraction of the received RF signal energy for information decoding (ID). Thus, according to \eqref{eq:PTSn1}, the harvested power at the ST$_{\rm{b}}$ is given by \cite{9,30}
\begin{equation}\label{eq:PH1}
\PH = \eta \beta \big(\underbrace{\PS ||\textbf{h}_{{{\rm{p}}\rm{s}}_{\rm{b}}}||^{2}}_{\text{RF EH}} + \underbrace{\PSb |\hSbSb|^2}_{\text{Self-EH}}\big),
\end{equation}
where $\eta$ denotes the energy conversion efficiency, $0 < \eta \leq 1$. The ST$_{\rm{b}}$ harvests both the dedicated energy from the PT and its own energy from the SI channel \cite{30,27}.

From \eqref{eq:PH1}, since the energy harvested from the PT is much larger than that of the receiver noise, we ignore the negligible energy harvested from the receiver noise \cite{9}. Next, the transmit power and the harvested power of the ST$_{\rm{b}}$ should satisfy $\PSb = \xi \psi \PH$ \cite{9}, where $\PSb$ refers to the transmit power of the ST$_{\rm{b}}$, $\xi$ denotes the portion of the harvested power consumed by the power amplifier of the ST$_{\rm{b}}$, and $\psi$ is the energy utilization efficiency, $0 \leq \xi \leq 1$ and $0 < \psi < 1$. Thus, the available transmit power of the ST$_{\rm{b}}$ is expressed as:
\begin{align}\label{eq:PSn}
\PSb &= \xi \psi \PH \nonumber \\ 
& = \eta \beta \xi \psi \big(\PS ||\textbf{h}_{{{\rm{p}}\rm{s}}_{\rm{b}}}||^{2} + \PSb |\hSbSb|^2\big) \nonumber \\
&= \rho \PS ||\textbf{h}_{{{\rm{p}}\rm{s}}_{\rm{b}}}||^{2},
\end{align} 
where $\rho = \frac{\eta \beta \xi \psi}{1 - \eta \beta \xi \psi |\hSbSb|^2}$.  

The sampled baseband signal at the ST$_{\rm{b}}$ is shown as:
\begin{equation}\label{eq:PTSn}
y^{\rm{ID}}_{\PT,\Sbb} = \hPSb \w \sqrt{(1-\beta)\PS} x_{0} + \hSbSb \sqrt{(1-\beta) \PSb} x_s + \sqrt{(1-\beta)} n_{\Sbb} + n_{c,\Sbb},
\end{equation}     
where $n_{c,\Sbb}$ denotes the AWGN circuit noise due to RF to baseband signal conversion. We ignore the antenna noise $n_{\Sbb}$ since its signal strength is much lower than that of the circuit processing noise $n_{c,\Sbb}$ \cite{9}. Hence, \eqref{eq:PTSn} can be rewritten as:
\begin{equation}\label{eq:PTSn2}
y^{\rm{ID}}_{\PT,\Sbb} = \hPSb \w \sqrt{(1-\beta)\PS} x_{0} + \hSbSb \sqrt{(1-\beta) \PSb} x_s   + n_{c,\Sbb}.
\end{equation}  

From \eqref{eq:PSn} and \eqref{eq:PTSn2}, the achievable data rate at the ST$_{\rm{b}}$ to decode $x_{0}$ is given by
\begin{align}\label{eq:RSbxo}
R_{{\Sbb},x_{0}} & = \log_{2} \left(1 + \frac{(1-\beta) \bar \gamma ||\textbf{h}_{{{\rm{p}}\rm{s}}_{\rm{b}}}||^{2}}{(1-\beta) \rho \bar \gamma |\hSbSb|^2 ||\textbf{h}_{{{\rm{p}}\rm{s}}_{\rm{b}}}||^{2}+1}\right),
\end{align}
where $\bar \gamma = \PS/\sigma_n$ denotes the transmit SNR.

The ST$_{\rm{b}}$ first decodes $x_{0}$. If $x_{0}$ is successfully decoded, the ST$_{\rm{b}}$ superimposes $x_{0}$ with $x_{i}$, $i=1,2,\ldots,\textit{M}$, according to the NOMA principle and broadcasts the composite signal $x_s=\sum\limits_{{i} = 0}^{M} \sqrt{\alpha_{i} \PSb} x_{i}$ to all $M+1$ receivers, where $x_{i}$ denotes the information signal intended to the PR (\textit{i} = 0) and SR$_{i}$ (\textit{i} $=1,\ldots,M$), and $\alpha_{i}$ denotes the corresponding power allocation coefficient with condition $\sum\limits_{{i} = 0}^{M} \alpha_{i} = 1$. Accordingly, the composite received signal at the receiver ${\rm{r}}_i$, $0\leq i \leq M$, from the ST$_{\rm{b}}$ can be expressed as:
\begin{equation}\label{eq:SnPR}
y_{{\rm{s_b}},\Ri} = \sum\limits_{{i} = 0}^{M} \sqrt{\alpha_{i} \PSb} \hSbRi x_{i} + n_{\Ri},
\end{equation}
where $h_{\rm{s}_{\rm{b}}r_0}$ is also denoted as $h_{\rm{s}_{\rm{b}}\rm{p}}$.

In this work, we consider the pessimistic case when the PR is grouped with \textit{M} SRs whose channel gains are stronger.\footnote{Note that this is a worst case for PR in terms of outage probability, system throughput, and sum transmission rate. This will be showed later in this paper. For the other cases when the channel gain of the PR is stronger than \textit{M} SRs, the analysis of these cases can be similarly obtained.} Hence, the channel gain of the ST$_{\rm{b}}$-PR link is always smaller than the channel gains of the ST$_{\rm{b}}$-SR$_{m}$ links. Henceforth, we assume that channel gains are sorted as $|h_{{\rm{s}}_{\rm{b}}{{\rm{p}}}}|^2 \leq |h_{{\rm{s}}_{\rm{b}}{{\rm{r}}_{1}}}|^2 \leq \ldots \leq |h_{{\rm{s}}_{\rm{b}}{{\rm{r}}_{M}}}|^2$ at the ST$_{\rm{b}}$ \cite{5,6}. Based on the NOMA principle, the power allocation coefficients used at the ST$_{\rm{b}}$ should fulfil the following condition $\alpha_{0} \geq \ldots \geq \alpha_{M}$ \cite{5,6}.

%

Each SR$_{m}$ performs SIC to distinguish the superimposed signals. The SR$_{m}$ first decodes $x_0$ followed by $x_1,\ldots,x_m$ according to the order of the ST$_{\rm{b}}$-SR$_{m}$ channel gains. Using \eqref{eq:PSn} and \eqref{eq:SnPR}, the achievable data rate at the SR$_{m}$ to decode $x_m$ is expressed as:
\begin{align}\label{eq:Rmxk}
R_{\text{r}_m}&=R_{\Rm,x_m} = \log_{2} \left(1 + {\rm{SINR}}_{\Rm,x_m} \right),
\end{align}
where SINR represents signal-to-interference-plus-noise ratio.

The rate \eqref{eq:Rmxk} is achievable provided that the condition $R_{\text{r}_v,x_m} \geq \bar{R}_m,\, \forall v > m,\, v=m+1, \ldots, M$, is met, where $\bar{R}_m$ denotes the predefined target data rate set for SR$_m$, and  $R_{\text{r}_v,x_m}$ is the achievable data rate at the SR$_v$ to decode $x_m$ with $v > m$, which is expressed as:
\begin{align}\label{eq:Rmxm}
R_{\text{r}_v,x_m} &= \log_{2} \left(1 +{\rm{SINR}}_{\text{r}_v,x_m} \right),
\end{align}    
where
\begin{align}\label{eq:Rmxk_SINR}
{\rm{SINR}}_{\text{r}_v,x_m} & = \log_{2} \left(1 + \frac{\alpha_{m} \rho \bar \gamma ||\textbf{h}_{{{\rm{p}}\rm{s}}_{\rm{b}}}||^{2} |h_{\text{s}_{\rm{b}}\text{r}_v}|^2}{\sum\limits_{{i} = m+1}^{M} \alpha_{i} \rho \bar \gamma ||\textbf{h}_{{{\rm{p}}\rm{s}}_{\rm{b}}}||^{2} |h_{\text{s}_{\rm{b}}\text{r}_v}|^2 + 1} \right).
\end{align}

The SR$_{v}$ then successfully subtracts $x_m$ from the received composite signal. The SIC process will last until its own signal $x_v$ is successfully decoded. Similarly, the achievable data rate at the PR to decode $x_0$ is expressed as:
\begin{align}\label{eq:R_P}
R_{\text{p}}&=R_{\text{r}_0,x_0}  = \log_{2} \left(1 + \frac{\alpha_{0} \rho \bar \gamma ||\textbf{h}_{{{\rm{p}}\rm{s}}_{\rm{b}}}||^{2} |h_{\rm{s}_{\rm{b}}\rm{p}}|^2}{\sum\limits_{{i} = 1}^{M} \alpha_{i} \rho \bar \gamma ||\textbf{h}_{{{\rm{p}}\rm{s}}_{\rm{b}}}||^{2} |h_{\rm{s}_{\rm{b}}\rm{p}}|^2 + 1} \right),
\end{align} 
where $h_{\rm{s}_{\rm{b}}\rm{p}} = h_{\rm{s}_{\rm{b}}r_0}$, and the achievable data rate at the SR$_M$ to decode $x_M$ is expressed as:
\begin{align}\label{eq:R_M}
R_{\text{r}_M}&= R_{\text{r}_M ,x_M}  = \log_{2} \left(1 + \alpha_{M} \rho \bar \gamma ||\textbf{h}_{{{\rm{p}}\rm{s}}_{\rm{b}}}||^{2} |h_{\text{s}_{\rm{b}}\text{r}_M}|^2 \right).
\end{align}
 
Note that if ST$_{\rm{b}}$ fails to decode $x_0$, it will not transmit any signals. This can be explained by the sole purpose that PT grants ST$_{\rm{b}}$ the permission to access primary spectrum only if ST$_{\rm{b}}$ is able to help it send its message to the PR. Thus, the achievable data rates at the PR and the SR$_{m}$ in both cases are shown as:
\begin{equation}
R_{\PR} = R_{\Rm} = 0.
\end{equation}

\section{Performance Analysis}\label{sec:per}

In this section, we will provide the performance analysis for primary and secondary networks in terms of outage probability and system throughput for the proposed system. The outage analysis is useful for practical purposes, as it offers information about the minimum data rate over which a wireless link or a communications system can operate as desired.

\subsection{Primary Network}

\subsubsection{Outage Probability}

In order to calculate the outage probability for the primary network, we need to obtain the cummulative distribution function (CDF) and probability density function (PDF) of $||\textbf{h}_{{\rm{p}}{\rm{s}}_{\rm{b}}}||^2$, which are given in Lemma~\ref{lem:CDF}.

\begin{lemma}\label{lem:CDF} The {\rm{CDF}} and {\rm{PDF}} of $||\rm{\bf{h}}_{{\rm{p}}{\rm{s}}_{\rm{b}}}||^2$ are derived as:
\begin{equation}\label{eq:Fhpsb}
\begin{split}
F_{||\rm{\bf{h}}_{{\rm{p}}{\rm{s}}_{\rm{b}}}||^2}(x) &= \sum\limits_{{l\rm{ = 0}}}^{K} {\left( \begin{array}{l}
	{K}\\
	{l}
	\end{array} \right)} (-1)^l \exp \left( -\frac{l x}{\lambda _{\rm{ps}}} \right) \sum\limits_{{j\rm{ = 0}}}^{l(N - 1)} \frac{C_{j} x^j}{{\lambda _{\rm{ps}}}^j},
\end{split}
\end{equation}  
\begin{equation}\label{eq:Phpsb}
\begin{split}
f_{||\rm{\bf{h}}_{{\rm{p}}{\rm{s}}_{\rm{b}}}||^2}(x) & = \sum\limits_{{l\rm{ = 1}}}^{K} {\left( \begin{array}{l}
	{K}\\
	{l}
	\end{array} \right)} {\left( { - 1} \right)^{{l}}} \exp \left( - \frac{l x}{\lambda _{\rm{ps}}} \right) \sum\limits_{{j\rm{ = 0}}}^{l(N - 1)} \frac{C_{j}}{{\lambda _{\rm{ps}}}^j} \left( j x^{j - 1} - \frac{l}{\lambda _{\rm{ps}}} x^{j} \right).
\end{split}
\end{equation}
\end{lemma}
\begin{IEEEproof}
See Appendix \ref{lemma1}.
\end{IEEEproof}

Next, we derive the outage probability for the primary network, where the outage event occurs under two circumstances. The first situation is when the ST$_{\rm{b}}$ can not successfully decode $x_{0}$. The second one occurs when the PR can not decode $x_{0}$ provided that the ST$_{\rm{b}}$ is able to decode $x_{0}$. Accordingly, the outage probability of the primary network is expressed as:
\begin{equation}\label{eq:PP}
P_{\rm{p}} = \Pr \left( R_{{\Sbb},x_{0}} < \bar{R}_0 \right) + \Pr\left( R_{{\Sbb},x_{0}} \geq \bar{R}_0, R_{\PR} < \bar{R}_0 \right),
\end{equation}  
where $\bar{R}_0$ is the target data rate of the primary signal $x_0$.

\begin{remark} Recalling \eqref{eq:PSn} and \eqref{eq:RSbxo}, it is obvious that when the power splitting (PS) coefficient $\beta$ increases, the transmit power ($\PSb$) and the achievable data rate ($R_{{\Sbb},x_{0}}$) of the {\rm{ST}}$_{\rm{b}}$ readily rises. It is also straightforward to see that the achievable data rates of the {\rm{PR}} and {\rm{SRs}}, which are illustrated in \eqref{eq:R_P} and \eqref{eq:R_M}, go up with the increase in $\beta$. Thus, the increase in $R_{{\Sbb},x_{0}}$ leads to the decrease in outage probability of the primary network in \eqref{eq:PP} and the secondary network in \eqref{eq:Prm}. However, when $\beta$ reaches the optimal value at which the lowest outage probability of primary and secondary networks is obtained, and if $\beta$ continues increasing to 1, the achievable data rate of {\rm{ST}}$_{\rm{b}}$ in \eqref{eq:RSbxo} decreases and reaches to 0. This is due to the fact that less power is left for the {\rm{ST}}$_{\rm{b}}$ to decode $x_0$; hence, the outage probability of both networks, as given in \eqref{eq:PP} and \eqref{eq:Prm}, increases and attains 1, which results in the worst performance. This result is also corroborated in {\rm{Fig. \ref{fig:6}}}.   
\end{remark}

\begin{theorem}
 The outage probability of the primary network can be approximated as:
\begin{align}\label{eq:PP1}
P_{\rm{p}} & \approx \sum\limits_{{l\rm{ = 0}}}^{K} {\left( \begin{array}{l}
	{K}\\
	{l}
	\end{array} \right)} (-1)^l \exp \left( -\frac{l \mu}{\lambda _{\rm{ps}} \bar \gamma} \right) \sum\limits_{{j\rm{ = 0}}}^{l(N - 1)} C_{j} \left(\frac{\mu}{{\lambda _{\rm{ps}}} \bar \gamma}\right)^{j} + \iota_{q} \widetilde \sum \Big(\sim\Big) \frac{(-1)^{c+n+l}}{q+c} \frac{C_{j}}{{l}^j} \nonumber \\
& \times \left[ \left(j + \frac{n \Theta_0 l}{\lambda _{\rm{sp}} \lambda _{\rm{ps}} \bar \gamma} \right) \Gamma\left( j, \frac{l \mu}{\lambda _{\rm{ps}} \bar\gamma}\right) - \Gamma\left( j+1, \frac{l \mu}{\lambda _{\rm{ps}} \bar\gamma} \right) - \frac{n \Theta_0 j l}{\lambda _{\rm{sp}} \lambda _{\rm{ps}} \bar \gamma} \Gamma\left(j-1, \frac{l \mu}{\lambda _{\rm{ps}} \bar\gamma}\right) \right],
\end{align} 
where $\iota_q = \frac{Q!}{(q-1)!(Q-q)!}$, $Q = M+1$, $\widetilde \sum = \sum\limits_{{c} = 0}^{Q-q} \sum\limits_{{n} = 0}^{q+c} \sum\limits_{{l\rm{ = 1}}}^{K} \sum\limits_{{j\rm{ = 0}}}^{l(N - 1)}$, and \hfill \\$\Big(\sim\Big) = {\left( \begin{array}{cc}
	{Q-q}\\
	{c}
	\end{array} \right)} {\left( \begin{array}{cc}
	{q+c}\\
	{n}
	\end{array} \right)} {\left( \begin{array}{l}
	{K}\\
	{l}
	\end{array} \right)}$. $\gamma_0 = 2^{\bar{R}_{0}} - 1$, $|\hSbSb|^2 = I_{\rm{SI}}$, $\mu = \frac{\gamma_0}{\left( 1 - \beta \right) \left( 1 - \gamma_0 \rho I_{\rm{SI}} \right)}$, $\Theta_0 = \frac{\gamma_0}{\rho \left( \alpha_0 - \sum\limits_{{i\rm{ = 1}}}^{M} \alpha_i \gamma_0\right)}$, and $\Gamma(\cdot,\cdot)$ is the incomplete gamma function [28, Eq.~(8.350.2)]. Note that $q = 1$ for the {\rm{PR}} and $q=m+1$ for the {\rm{SR}}$_m$. $P_{\rm{p}}$ is derived in \eqref{eq:PP1} when $\gamma_{0} < \min \left[ 1/(\rho I_{\rm{SI}}), \alpha_0/\left(\sum\limits_{{i} = 1}^{M} \alpha_i\right) \right]$, otherwise $P_{\rm{p}} = 1$.
\end{theorem}
\begin{IEEEproof}
See Appendix \ref{theorem1}.
\end{IEEEproof}

\subsubsection{Throughput}

Given that the ST$_{\rm{b}}$ transmits information to the PR at a constant rate $\bar{R}_0$ bps/Hz, the throughput of the primary network in the delay-limited transmission mode is computed as \cite{21}:

\begin{equation}\label{eq:tau}
\nu_{\rm{p}} = \left(1 - P_{\rm{p}}\right) \bar{R}_0,
\end{equation}   
where $P_{\rm{p}}$ is shown in \eqref{eq:PP1}.

\begin{remark} From \eqref{eq:tau}, since $P_{\rm{p}}$ is smaller than or equal to 1 and as the transmit {\rm{SNR}} $\bar \gamma$ increases, the outage performance of the primary network greatly improves ($P_{\rm{p}}$ decreases as shown in {\rm{Fig. \ref{fig:3}}}). Hence, the throughput of the primary network $\nu_{\rm{p}}$ goes up with the rise in the transmit {\rm{SNR}} $\bar \gamma$ and converges to a throughput floor which equals to  $\bar{R}_0$ at high $\bar \gamma$. Note that the throughput of the secondary network $\nu_{\rm{s}}$ given in \eqref{eq:taus} also enhances with the increase in the transmit {\rm{SNR}} $\bar \gamma$ and attains a throughput floor which equals to $\sum\limits_{{m\rm{ = 1}}}^{M} \bar{R}_m$ at the high $\bar \gamma$. These results are also verified in {\rm{Fig. \ref{fig:15}}}.
\end{remark}

\subsection{Secondary Network}

\subsubsection{Outage Probability}

In this section, we derive the outage probability for the secondary network. The outage event of the SR$_{m}$ occurs either when the ST$_{\rm{b}}$ can not successfully decode $x_{0}$ or when the SR$_{m}$ fails to decode any $x_{m'}$, $0 \leq m' \leq m$, as long as ST$_{\rm{b}}$ successfully decodes $x_{0}$. Based on this, the outage probability of the SR$_{m}$ is shown as:
\begin{equation}\label{eq:Prm}
P_{{\rm{r}}_m} = \Pr \left( R_{{\Sbb},x_{0}} < \bar{R}_0 \right) + \Pr\left( R_{{\Sbb},x_{0}} \geq \bar{R}_0, \overline{P}_{{\rm{r}}_m} \right),
\end{equation} 
where $\overline{P}_{{\rm{r}}_m}$ denotes the outage probability that the SR$_{m}$ fails to decode any $x_{m'}$, $0 \leq m' \leq m$.

\begin{theorem} The outage probability of the {\rm{SR}}$_{m}$ can be approximated as:
\begin{align}\label{eq:Prm2}
P_{{\rm{r}}_{m}} & \approx \sum\limits_{{l\rm{ = 0}}}^{K} {\left( \begin{array}{l}
{K}\\
{l}
\end{array} \right)} (-1)^l \exp \left( -\frac{l \mu}{\lambda _{\rm{ps}} \bar \gamma} \right) \sum\limits_{{j\rm{ = 0}}}^{l(N - 1)} C_{j} \left(\frac{\mu}{{\lambda _{\rm{ps}}} \bar \gamma}\right)^{j} + \iota_{q} \widetilde \sum \Big(\sim\Big) \frac{(-1)^{c+n+l}}{q+c} \frac{C_{j}}{{l}^j}  \nonumber \\
&\times \left[ \left(j + \frac{n \Theta l}{\lambda _{\rm{sr}} \lambda _{\rm{ps}} \bar \gamma} \right) \Gamma\left( j, \frac{l \mu}{\lambda _{\rm{ps}} \bar\gamma}\right) - \Gamma\left( j+1, \frac{l \mu}{\lambda _{\rm{ps}} \bar\gamma} \right) - \frac{n \Theta j l}{\lambda _{\rm{sr}} \lambda _{\rm{ps}} \bar \gamma} \Gamma\left(j-1, \frac{l \mu}{\lambda _{\rm{ps}} \bar\gamma}\right) \right],
\end{align}
where $\gamma_{m'} = 2^{\bar{R}_{m'}} - 1$, $\bar{R}_{m'}$ denotes the target data rate of $x_{m'}$, $\Theta = \max \left( \Theta_{0},\ldots,\Theta_{m} \right)$, $\Theta_{m'} = \frac{\gamma_{m'}}{\rho \left( \alpha_{m'} - \sum\limits_{{i = m'+1}}^{M} \alpha_i \gamma_{m'} \right)}$, $0 \leq m' \leq m$. Note that, $q = m + 1$ for {\rm{SR}}$_m$.
\end{theorem}
\begin{IEEEproof}
See Appendix \ref{theorem2}.
\end{IEEEproof}
\subsubsection{Throughput}

The system throughput of the secondary network is given by:

\begin{equation}\label{eq:taus}
\nu_{\rm{s}} = \sum\limits_{{m\rm{ = 1}}}^{M} \left(1 - P_{{\rm{r}}_m}\right) \bar{R}_m,
\end{equation} 
where $P_{{\rm{r}}_m}$ is obtained from \eqref{eq:Prm2}.

\section{Sum Rate Maximization}\label{sec:sum}

In the previous section, the power allocation coefficients were kept fixed, which are not optimal. Thus, in this section, to further improve the performance of the system, we formulate the sum rate maximization problem to obtain the optimal power allocation coefficients at the ST$_{\rm{b}}$ and develop a solution method and algorithm.

\subsection{Problem Formulation}

Recalling \eqref {eq:Rmxk} and the condition $R_{\text{r}_v,x_m} \geq \bar{R}_m,\, \forall v > m,\, v=m+1, \ldots, M$, the achievable data rate to decode $x_m$ at the SR$_m$ is expressed as:
\begin{equation}\label{eq:x0}
R_{x_m} = \log_{2} \left(1 + \min\left( {\rm{SINR}}_{{\rm{r}}_m,x_m}, \ldots, {\rm{SINR}}_{{\rm{r}}_M,x_m} \right) \right).
\end{equation}

\begin{lemma}\label{lem:rate} From \eqref{eq:x0}, the achievable data rate to decode $x_m$ by the {\rm{SR}}$_m$ is rewritten as
\begin{equation}\label{eq:x01}
R_{x_m} = \log_{2} \left(1 + {\rm{SINR}}_{{\rm{r}}_m,x_m} \right).
\end{equation}
\end{lemma}
\begin{IEEEproof}\label{eq:x02}
We prove this by contradiction. Let us first assume that ${\rm{SINR}}_{\text{r}_0,x_0} \geq {\rm{SINR}}_{{\rm{r}}_1,x_0}$. Then, we have
\begin{align}
{\rm{SINR}}_{\text{r}_0,x_0}& \geq {\rm{SINR}}_{{\rm{r}}_1,x_0} \nonumber \\
&\Leftrightarrow \frac{\alpha_{0} \rho \bar \gamma ||\textbf{h}_{{{\rm{p}}\rm{s}}_{\rm{b}}}||^{2} \gSbP}{\sum\limits_{{i} = 1}^{M} \alpha_{i} \rho \bar \gamma ||\textbf{h}_{{{\rm{p}}\rm{s}}_{\rm{b}}}||^{2} \gSbP + 1} \geq \frac{\alpha_{0} \rho \bar \gamma ||\textbf{h}_{{{\rm{p}}\rm{s}}_{\rm{b}}}||^{2} |h_{\rm{s_b} {\rm{r}_\textit{1}}}|^2}{\sum\limits_{{i} = 1}^{M} \alpha_{i} \rho \bar \gamma ||\textbf{h}_{{{\rm{p}}\rm{s}}_{\rm{b}}}||^{2} |h_{\rm{s_b} {\rm{r}_\textit{1}}}|^2 + 1} \nonumber \\
&\Leftrightarrow \alpha_{0} \sum\limits_{{i} = 1}^{M} \alpha_{i} \left( \rho \bar \gamma ||\textbf{h}_{{{\rm{p}}\rm{s}}_{\rm{b}}}||^{2} \right)^2 \gSbP |h_{\rm{s_b} {\rm{r}_\textit{1}}}|^2 + \alpha_{0} \rho \bar \gamma ||\textbf{h}_{{{\rm{p}}\rm{s}}_{\rm{b}}}||^{2} \gSbP \nonumber \\
& \geq \alpha_{0} \sum\limits_{{i} = 1}^{M} \alpha_{i} \left( \rho \bar \gamma ||\textbf{h}_{{{\rm{p}}\rm{s}}_{\rm{b}}}||^{2} \right)^2 \gSbP |h_{\rm{s_b} {\rm{r}_\textit{1}}}|^2 + \alpha_{0} \rho \bar \gamma ||\textbf{h}_{{{\rm{p}}\rm{s}}_{\rm{b}}}||^{2} |h_{\rm{s_b} {\rm{r}_\textit{1}}}|^2 \nonumber \\
&\Leftrightarrow \gSbP \geq |h_{\rm{s_b} {\rm{r}_\textit{1}}}|^2.\nonumber
\end{align}

The above result violates the condition of channel gain order that we set earlier. Hence, ${\rm{SINR}}_{\text{r}_0,x_0} \leq {\rm{SINR}}_{{\rm{r}}_1,x_0}$. Following similar derivation steps as previously mentioned, we have ${\rm{SINR}}_{\text{r}_0,x_0} \leq {\rm{SINR}}_{\Rm,x_0}, \forall m$. Thus, $\min\left( {\rm{SINR}}_{\text{r}_0,x_0}, {\rm{SINR}}_{{\rm{r}}_\textit{1},x_0}, \ldots, {\rm{SINR}}_{{\rm{r}}_M,x_0} \right) = {\rm{SINR}}_{\text{r}_0,x_0}$. 

By generalizing the above result, the achievable data rate to decode the SR$_{m}$'s data $x_m$, $1 \leq m \leq M$, by all SR$_{m'}$, $m \leq m' \leq M$, is given by
\begin{align}\label{xm}
R_{x_m} &= \log_{2} \left(1 + \min ({\rm{SINR}}_{{\rm{r}}_m,x_m},\ldots,{\rm{SINR}}_{{\rm{r}}_M,x_m}) \right) \nonumber \\
& = \log_{2} \left(1 + {\rm{SINR}}_{{\rm{r}}_m,x_m} \right).
\end{align}
\end{IEEEproof} 

Using Lemma~\ref{lem:rate}, the sum rate maximization problem is formulated as follows:


\begin{subequations}\label{eq:max}
\begin{align}
\underset{\boldsymbol{\alpha}}{\text{maximize}} \label{eq:op1}
& \qquad R_{x_0}  + \sum\limits_{{m'\rm{ = 1}}}^{M-1} R_{x_{m'}} + R_{x_M} \\ \label{eq:op2}
\text{s.t.} 
&  \qquad R_{x_i} \geq \bar{R}_i, 0 \leq i \leq M, \\ \label{eq:op3}
&  \qquad\alpha_{0} \geq \ldots \geq \alpha_{i} \geq \ldots \geq \alpha_{M}, \\ \label{eq:op4}
 & \qquad \sum\limits_{{i\rm{ = 0}}}^{M} \alpha_i \leq 1,
\end{align}
\end{subequations}
where $\boldsymbol{\alpha} = \left[\alpha_{0}, \alpha_{1},\ldots, \alpha_{M}\right]$ denotes the power allocation coefficients vector. The constraint \eqref{eq:op2} ensures that the quality-of-service (QoS) requirement $\bar{R}_i$ of each link is guaranteed, the constraint in \eqref{eq:op3} represents the necessary conditions related to fairness among the users, and the constraint in \eqref{eq:op4} puts limit on the total transmit power at the ST$_{\rm{b}}$.

We can observe that the original problem \eqref{eq:max} is non-convex because of the objective function, and it is difficult to solve it quickly to obtain a global solution. Instead, we aim to solve it for suboptimal solution with faster convergence. To this end, we invoke several useful steps, such as its equivalent transformation and approximation, as described in the next subsection, to solve \eqref{eq:max}. 

\subsection{Proposed Solution}

By noting that the logarithmic function is a monotonically non-decreasing function, \eqref{eq:max} is equivalently written as:
\begin{equation}\label{eq:max1}
\underset{\boldsymbol{\alpha}}{\text{maximize}} \Bigg\{\prod_{i=0}^M(1+\text{SINR}_{\text{r}_i,x_i})\mid \eqref{eq:op2}, \eqref{eq:op3},\,\text{and}\,\eqref{eq:op4}\Bigg\}.
\end{equation}

By introducing a new vector of slack variables $\mathbf{t}=[t_0, \ldots,t_M]$,  the problem \eqref{eq:max1} can be equivalently recast as:

\begin{subequations}\label{eq:max2}
\begin{align}
& \underset{\boldsymbol{\alpha},\mathbf{t}}{\text{maximize}} \label{eq:op12}
& & \prod_{i=0}^{M} t_i \\ \label{eq:op22}
& \text{s.t.} 
& & \text{SINR}_{r_0,x_0} \geq t_0 - 1,\\ \label{eq:op32}
& & & \text{SINR}_{r_{m'},x_{m'}} \geq t_{m'} - 1, \qquad 1 \leq m' \leq M-1, \\ \label{eq:op42}	
& & & \text{SINR}_{r_M,x_M} \geq t_M - 1, \\ \label{eq:op52}
& & & \alpha_{0} \rho \bar \gamma ||\textbf{h}_{{{\rm{p}}\rm{s}}_{\rm{b}}}||^{2} \gSbP \geq \left( 2^{\bar{R}_0} - 1 \right) \left(\sum\limits_{{i} = 1}^{M} \alpha_{i} \rho \bar \gamma ||\textbf{h}_{{{\rm{p}}\rm{s}}_{\rm{b}}}||^{2} \gSbP + 1\right), \\ \label{eq:op62}
& & & \alpha_{m'} \rho \bar \gamma ||\textbf{h}_{{{\rm{p}}\rm{s}}_{\rm{b}}}||^{2} |h_{\rm{s_b} {\rm{r}_{m'}}}|^2 \geq \left(2^{\bar{R}_{m'}} - 1\right) \left( \sum\limits_{{i} = m'+1}^{M} \alpha_{i} \rho \bar \gamma ||\textbf{h}_{{{\rm{p}}\rm{s}}_{\rm{b}}}||^{2} |h_{\rm{s_b} {\rm{r}_{m'}}}|^2 + 1  \right), 1 \leq m' \leq M-1, \\ \label{eq:op72}
& & & \alpha_{M} \rho \bar \gamma ||\textbf{h}_{{{\rm{p}}\rm{s}}_{\rm{b}}}||^{2} |h_{\rm{s_b} {\rm{r}_\textit{M}}}|^2 \geq \left(2^{\bar{R}_M} - 1\right), \\ \label{eq:op82}
& & & \eqref{eq:op3} \ \& \ \eqref{eq:op4}.
\end{align}
\end{subequations}

By noting that the constraints \eqref{eq:op22}-\eqref{eq:op42} are active at optimality, then \eqref{eq:max2} is the equivalent formulation of \eqref{eq:max1}. Further, note that the objective function \eqref{eq:op12} is the product of optimization variables $t_i, \forall i$, and hence, admits a second-order cone (SOC) representation. However, \eqref{eq:max2} is still intractable because of the non-convexity involved in constraints \eqref{eq:op22} and \eqref{eq:op32}. Next, we consider the constraints \eqref{eq:op22} and \eqref{eq:op32}, and reformulate them as follows:
\begin{subnumcases}
{\label{eq:op22r} \eqref{eq:op22} \Leftrightarrow} 
\sum\limits_{{i} = 1}^{M} \alpha_{i} \rho \bar \gamma ||\textbf{h}_{{{\rm{p}}\rm{s}}_{\rm{b}}}||^{2} \gSbP + 1 \leq z_{0},\IEEEyessubnumber\label{eq:op22rc}\\
\alpha_{0} \rho \bar \gamma ||\textbf{h}_{{{\rm{p}}\rm{s}}_{\rm{b}}}||^{2} \gSbP \geq z_{0} t_0 - z_{0},\label{eq:op22re}
\end{subnumcases}
\begin{subnumcases}
{\label{eq:op32r} \eqref{eq:op32} \Leftrightarrow} 
\sum\limits_{{i} = m'+1}^{M} \alpha_{i} \rho \bar \gamma ||\textbf{h}_{{{\rm{p}}\rm{s}}_{\rm{b}}}||^{2} |h_{\rm{s_b} {\rm{r}_{m'}}}|^2 + 1 \leq z_{m'},\IEEEyessubnumber\label{eq:op32rc}\\
\alpha_{m'} \rho \bar \gamma ||\textbf{h}_{{{\rm{p}}\rm{s}}_{\rm{b}}}||^{2} |h_{\rm{s_b} {\rm{r}_{m'}}}|^2 \geq z_{m'} t_{m'} - z_{m'},\label{eq:op32re}
\end{subnumcases}
where $1 \leq m' \leq M-1$ and $\mathbf{z} = \left[z_{0}, z_{1},\ldots, z_{M-1}\right]$ represent newly introduced variables. After replacing \eqref{eq:op22} and \eqref{eq:op32} with \eqref{eq:op22rc} and \eqref{eq:op22re}, and \eqref{eq:op32rc} and \eqref{eq:op32re}, respectively, we get an equivalent formulation of \eqref{eq:max2}. However, this is still non-convex because of \eqref{eq:op22re} and \eqref{eq:op32re}. To this end, we approximate them by using the first-order Taylor series expansion. Firstly, we consider \eqref{eq:op22re} and rewrite the multiplicative terms $z_{0} t_0$ in the form of the difference of convex (d.c.) functions as follows:
\begin{align}
z_{0} t_0 = \frac{1}{4} \left[ \left(z_{0} + t_0\right)^2 - \left(z_{0} - t_0\right)^2 \right].
\end{align} 

Then, approximating $(z_0-t_0)^2$ by using the first-order Taylor series around the point $(z^{(\tau)}_{0},t^{(\tau)}_0)$, which is obtained at the $\tau$th iteration, \eqref{eq:op22re} can be replaced with the following convex constraint:
\begin{align}\label{eq:1}
\alpha_{0} \rho \bar \gamma ||\textbf{h}_{{{\rm{p}}\rm{s}}_{\rm{b}}}||^{2} \gSbP &\geq 0.25 \left(z_{0} + t_0\right)^2 - z_{0} - 0.25 \bigg[ \left(z^{(\tau)}_{0} - t^{(\tau)}_0\right)^2 + 2 \left(z^{(\tau)}_{0} - t^{(\tau)}_0\right) \nonumber \\
&\times \left(z_{0} - z^{(\tau)}_{0} - t_0 + t^{(\tau)}_0\right) \bigg].
\end{align}

Similarly, \eqref{eq:op32re} can be replaced with the following convex constraint:
\begin{align}\label{eq:2}
\alpha_{m'} \rho \bar \gamma ||\textbf{h}_{{{\rm{p}}\rm{s}}_{\rm{b}}}||^{2} |h_{\rm{s_b} {\rm{r}_{m'}}}|^2 &\geq 0.25 \left(z_{m'} + t_{m'}\right)^2 - z_{m'} - 0.25 \bigg[ \left(z^{(\tau)}_{m'} - t^{(\tau)}_{m'}\right)^2 + 2 \left(z^{(\tau)}_{m'} - t^{(\tau)}_{m'}\right) \nonumber \\ 
&\times \left(z_{m'} - z^{(\tau)}_{m'} - t_{m'} + t^{(\tau)}_{m'}\right) \bigg].
\end{align}

Finally, the convex problem to be solved at the $\tau$th iteration can be written as:
\begin{subequations}\label{eq:max4}
\begin{align}
& \underset{\boldsymbol{\alpha}, \mathbf{t}, \mathbf{z}}{\text{maximize}} \label{eq:op14}
& & \prod_{i=0}^{M} t_i \\ \label{eq:op24}
& \text{s.t.} 
& & \begin{cases}
\sum\limits_{{i} = 1}^{M} \alpha_{i} \rho \bar \gamma ||\textbf{h}_{{{\rm{p}}\rm{s}}_{\rm{b}}}||^{2} \gSbP + 1 \leq z_{0},\\
\eqref{eq:1},
\end{cases},\\ \label{eq:op34}
& & & \begin{cases}
\sum\limits_{{i} = m'+1}^{M} \alpha_{i} \rho \bar \gamma ||\textbf{h}_{{{\rm{p}}\rm{s}}_{\rm{b}}}||^{2} |h_{\rm{s_b} {\rm{r}_{m'}}}|^2 + 1 \leq z_{m'},\\
\eqref{eq:2},
\end{cases}, 1 \leq m' \leq M-1, \\ \label{eq:op44}	
& & & \alpha_{M} \rho \bar \gamma ||\textbf{h}_{{{\rm{p}}\rm{s}}_{\rm{b}}}||^{2} |h_{\rm{s_b} {\rm{r}_\textit{M}}}|^2 \geq t_M - 1, \\ \label{eq:op54}
& & & \eqref{eq:op52}-\eqref{eq:op82}.
\end{align}
\end{subequations}

After solving \eqref{eq:max4}, we update the involved optimization variables and repeat the procedure until convergence. The proposed algorithm is summarized in Algorithm \ref{alg_IterativeAlgorithm}.

\begin{algorithm}[!h]
	\begin{algorithmic}[1]
		\protect\caption{{\color{black}Proposed iterative algorithm to solve \eqref{eq:max}}}
		\label{alg_IterativeAlgorithm}
		\global\long\def\algorithmicrequire{\textbf{Initialization:}}
		\REQUIRE  Set $\tau=0$ and generate an initial feasible point $(t^{(0)}_i, z^{(0)}_{m'}), 0 \leq i \leq M, 0 \leq m' \leq M-1$.
		\REPEAT
		\STATE Solve the convex program \eqref{eq:max4} to obtain the optimal solution: $(t^{(\tau),\star}_i, z^{(\tau),\star}_{m'}), 0 \leq i \leq M, 0 \leq m' \leq M-1$.
		\STATE Update $(t^{(\tau+1)}_i, z^{(\tau+1)}_{m'}) := (t^{(\tau),\star}_i, z^{(\tau),\star}_{m'}), 0 \leq i \leq M, 0 \leq m' \leq M-1.$
		\STATE Set $\tau=\tau+1.$
		\UNTIL Convergence\\
	\end{algorithmic} 
\end{algorithm}

\subsection{Convergence and Complexity Analysis}
	
The proposed algorithm  begins with a random initial feasible point for the updated variables $(t^{(0)}_i, z^{(0)}_{m'}), 0 \leq i \leq M, 0 \leq m' \leq M-1$. In each iteration, we solve the convex program \eqref{eq:max4} to produce the next feasible point $(t^{(\tau+1)}_i, z^{(\tau+1)}_{m'}), 0 \leq i \leq M, 0 \leq m' \leq M-1$. This procedure is successively repeated until convergence, which is stated in the following proposition.	
  	
\begin{proposition}\label{pro:1} Initialized from a feasible point $(t^{(0)}_i, z^{(0)}_{m'}), 0 \leq i \leq M, 0 \leq m' \leq M-1$, Algorithm 1 produces a sequence $(t^{(\tau)}_i, z^{(\tau)}_{m'}), 0 \leq i \leq M, 0 \leq m' \leq M-1,$ of improved solutions to problem \eqref{eq:max4}, which satisfy the Karush-Kuhn-Tucker (KKT) conditions. The sequence  $\Bigl\{\prod_{i=0}^{M} t^{\left(\tau\right)}_i \Bigr\}_{\tau=1}^{\infty}$ is monotonically increasing and converges after a finite number of iterations for a given error tolerance $\epsilon >0$.
\end{proposition}

\begin{proof}
See Appendix \ref{proposition 1}.
\end{proof}

An SOC programming (SOCP) is solved in each iteration of the procedure illustrated in Algorithm 1. Hence, the worst case of the complexity is regulated by the SOCP in each run. To assess the complexity estimate, the worst case complexity of the SOCP in \eqref{eq:max4} is estimated. As shown in \cite{1op}, for general interior-point methods, the complexity of the SOCP relies on the number of constraints, variables, and the dimension of each SOC constraint. The total number of constraints in \eqref{eq:max4} is $3 M + 4 + a$, where $a$ is a non-negative integer constant and denotes the SOC constraints with different $M$. This is because the objective function in \eqref{eq:max4} represents the equivalent SOC of the geometric mean \cite{1op}. Hence, the number of iterations required to decrease the duality gap to a small constant is upper bounded by $\mathcal{O} \left(\sqrt{3 M + 4 + a}\right)$ \cite{1op}. The per iteration worst case complexity estimate of the interior-point method is $\mathcal{O}\left(\left(3 M + 2 + a\right)^2 \left(3M\right)\right)$, where $3 M + 2 + a$ and $3M$ represent the number of optimization variables and the dimension of the SOC constraints in \eqref{eq:max4}, respectively.

\section{Numerical Results} \label{sec}

In this section, simulation results are presented to verify the findings presented in Sections \ref{sec:per} and \ref{sec:sum}. Without loss of generality, we set $M = 2$, $\lambda _{\rm{ps}} = 5$, and $\lambda _{\rm{sp}} = \lambda _{\rm{sr}} = 50$.{\footnote{We consider the network scenario where STs are located near the PR and SRs, but they are far from PT.}} The target data rates of primary and secondary signals are $\bar{R}_0 = \bar{R}_1 = \bar{R}_2 = 0.5$ bps/Hz. The residual SI channel $\hSbSb$ is modeled as described in \cite{1,12}, and its variance is set to $|\hSbSb|^2 = I_{\rm{SI}} = \sqrt{\zeta}$, where $\zeta=$ -1 dB \cite{1}. The energy conversion efficiency is set to be $\eta = 0.75$, $\xi = 1$, and $\psi = 0.75$ \cite{9}. For the results corresponding to the outage probability analysis, the power allocation coefficients are set as $\alpha_0 = 0.6$, $\alpha_1 = 0.3$, and $\alpha_2 = 0.1$.
	
\subsection{Outage Probability}
	
\begin{figure}[!h]
\begin{center}
\epsfxsize=9cm \epsfbox{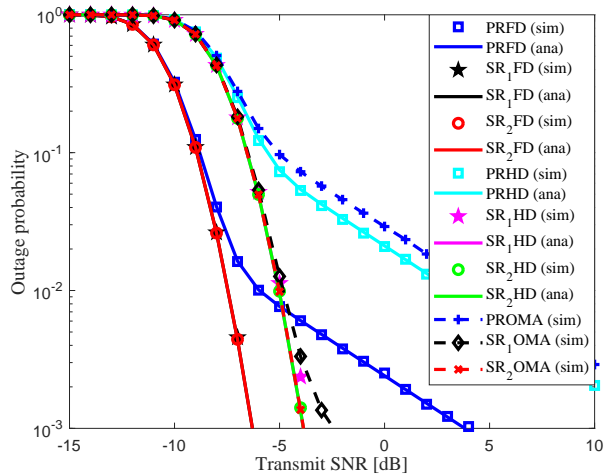} \caption{Outage probability of primary and secondary networks operating in FD and HD modes, where $\beta = 0.8$, $N = 5$, and $K = 3$. ana: analytical results; sim: simulation results.} \label{fig:3}
\end{center}
\end{figure}
	
In Fig. \ref{fig:3}, we compare the outage probability performance between primary and secondary networks operating in FD and HD modes. We consider the conventional OMA-TDMA scheme as a benchmark. In OMA-TDMA scheme, the first $\kappa$ fraction of the block time $T$, where $\kappa$ denotes the time allocation parameter, is used for the transmission from the PT to the ST$_{\rm{b}}$, whereas the remaining (1-$\kappa$) fraction of $T$ is equally divided into ($M$+1) time slots for the transmission from the ST$_{\rm{b}}$ to the PR and $M$ SRs. In this figure, we set the block time allocation parameter $\kappa$ to 1/2. Further, PRFD, SR$_{1}$FD, and SR$_{2}$FD curves denote the outage performance of the PR, the SR$_{1}$, and the SR$_{2}$ when the $\Sb$ operates in FD mode, respectively. PRHD, SR$_{1}$HD, and SR$_{2}$HD curves represent the outage performance of the PR, the SR$_{1}$, and the SR$_{2}$ when the $\Sb$ operates in HD mode, respectively, while PROMA, SR$_{1}$OMA, and SR$_{2}$OMA curves are the outage performance of the PR, the SR$_{1}$, and the SR$_{2}$ in OMA-TDMA scheme, respectively. As can be seen from Fig. \ref{fig:3}, for both FD and HD cases, the outage performance of the primary network is worse than that of the secondary network\footnote{Note that, if the PR is grouped with SRs whose channel gains are weaker, the outage performance of the PR and SRs will swap.}. This can be explained that even though the PR is assigned larger power allocation coefficient compared to the SR$_{1}$ and the SR$_{2}$, the quality of ST$_{\rm{b}}$-PR channel is worst compared to that of ST$_{\rm{b}}$-SR$_{1}$ and ST$_{\rm{b}}$-SR$_{2}$ channels and since the PR suffers interference from SRs when it decodes its information, which lets the achievable data rate of the PR in \eqref{eq:R_P} be smaller than that of SRs in \eqref{eq:Rmxk}. Hence, outage probability obtained at the PR is higher than that at SRs. In addition, we can see that outage probability of both networks in FD case is much smaller than that in HD mode, confirming the benefit of FD compared with HD technique. This makes sense since with FD, the ST$_{\rm{b}}$ can simultaneously receive the signal from the PT and forward its signal to the PR and SRs, whereas with the HD technique the ST$_{\rm{b}}$ has to separate the time used for receiving the PT's signal and the time for its transmission to the PR and SRs. Furthermore, the performance of our proposed FD-NOMA scheme is obviously better than OMA-TDMA scheme for both primary and secondary networks since our proposed scheme enhances the use of primary spectrum resource, i.e., the PR and SRs are concurrently served in the same resource block, while the OMA-TDMA scheme requires separated resource blocks. From Fig. \ref{fig:3}, we can see that the tight closed-form approximate expression curves, which are shown in \eqref{eq:PP1} and \eqref{eq:Prm2}, match well with simulation results. Hence, this verifies the correctness of the mathematical analysis.

\begin{figure}[!h]
	\begin{minipage}{0.46\linewidth}
		\begin{center}
			\epsfxsize=8cm \epsfbox{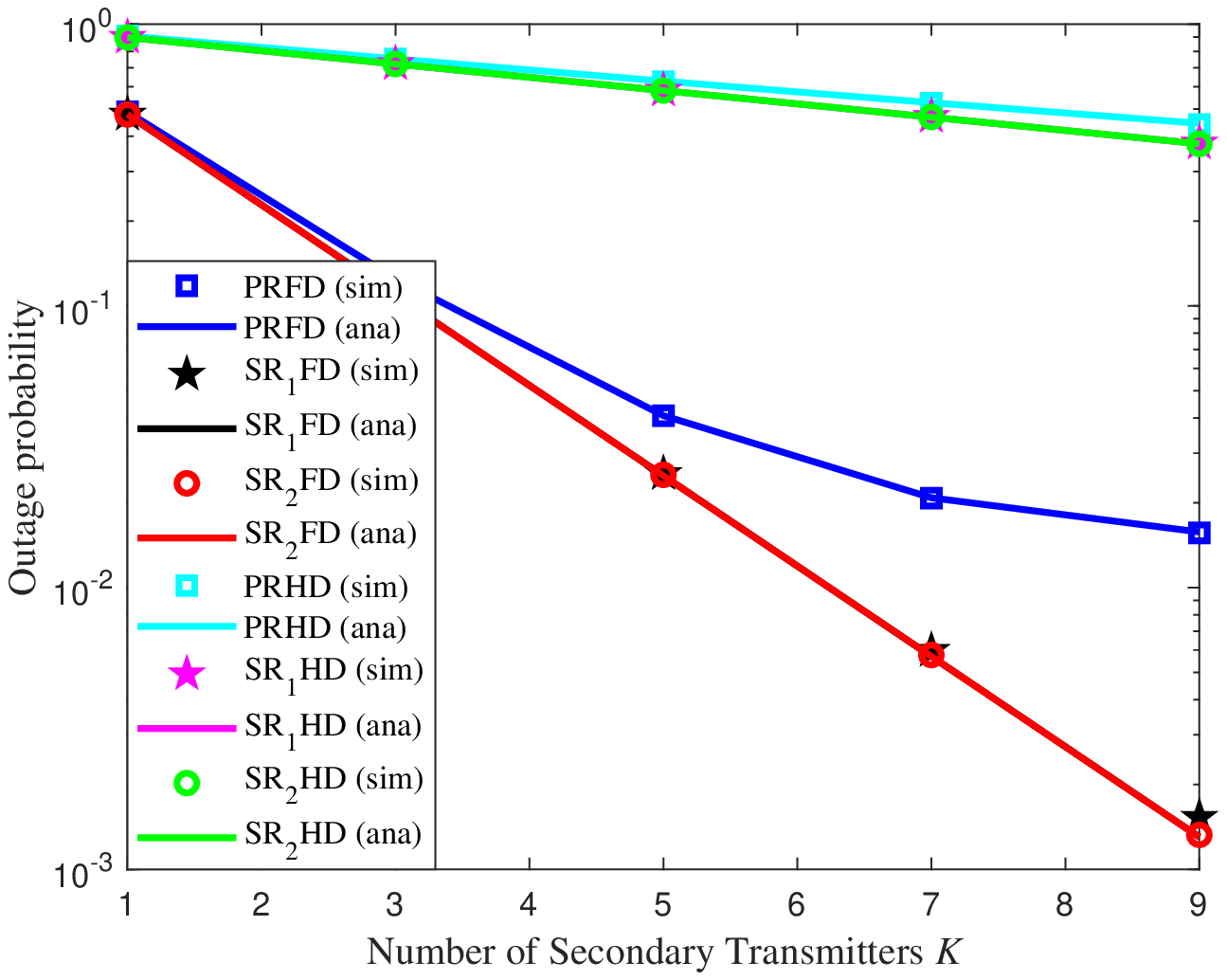} \caption{Impact of the number of STs $K$ on the outage performance of primary and secondary networks operating in FD and HD modes, where $\beta = 0.8$, $N = 5$, and $\bar \gamma = -9$ dB.} \label{fig:4}
		\end{center}
	\end{minipage}
	\hfill
	\begin{minipage}{0.46\linewidth}
		\begin{center}
			\epsfxsize=8cm \epsfbox{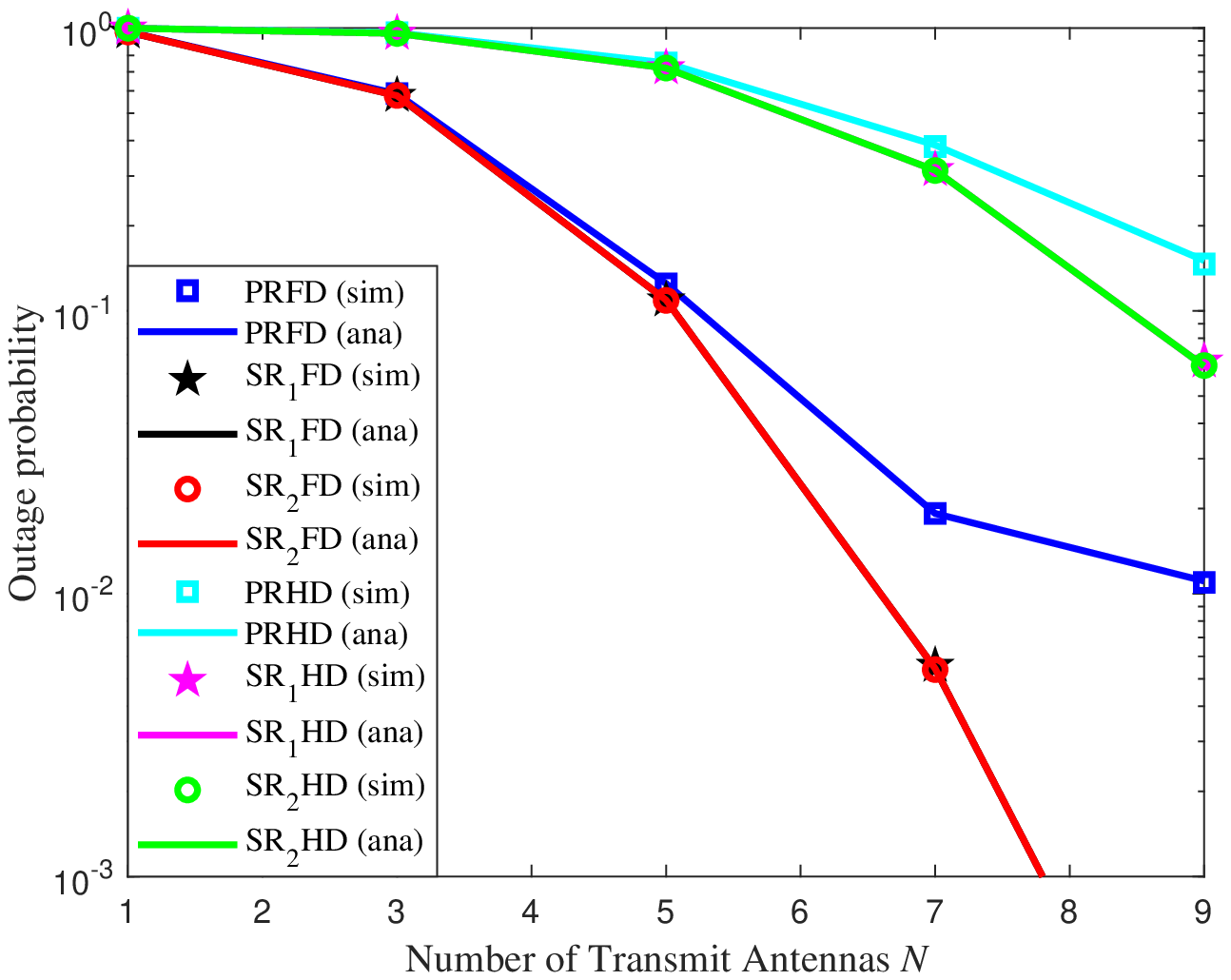} \caption{Effect of the number of transmit antennas $N$ on the outage performance of primary and secondary networks operating in FD and HD modes, where $\beta = 0.8$, $K = 3$, and $\bar \gamma = -9$ dB.} \label{fig:5}
		\end{center}
	\end{minipage}
\end{figure}
	
Fig. \ref{fig:4} shows the impact of the number of STs $K$ on the outage probability of primary and secondary networks operating in FD and HD modes. We observe that when $K$ increases, the outage performance of both primary and secondary networks in FD and HD modes is greatly improved. This is because the growth in the number of STs increases the probability of choosing the optimal ST$_{\rm{b}}$.  
	
Next, the effect of the number of transmit antennas $N$ on the outage probability of primary and secondary networks operating in FD and HD modes is shown in Fig. \ref{fig:5}. We can see that for both FD and HD cases, the larger the number of transmit antennas at the PT, the better the outage performance can be obtained. The reason is that a larger number of antennas provides higher spatial diversity, which in turn leads to improved signal reception quality at the ST$_{\rm{b}}$. 

\begin{figure}[!h]
	\begin{minipage}{0.46\linewidth}
		\begin{center}
			\epsfxsize=8cm \epsfbox{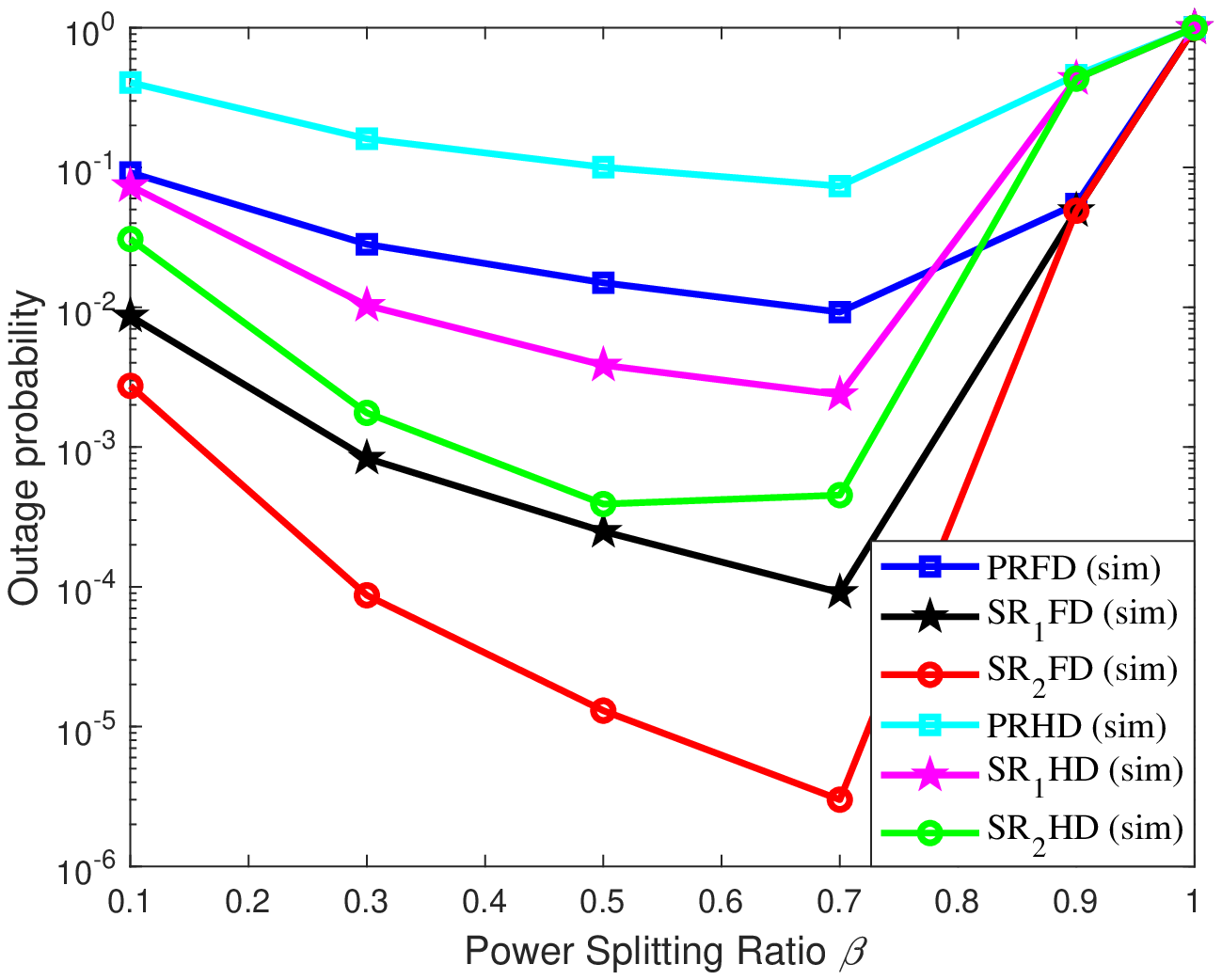} \caption{Impact of PS $\beta$ on the outage performance of primary and secondary networks operating in FD and HD modes, where $N = 5$, $K = 3$, and $\bar \gamma = -5$ dB.} \label{fig:6}
		\end{center}
	\end{minipage}
	\hfill
	\begin{minipage}{0.46\linewidth}
		\begin{center}
			\epsfxsize=8cm \epsfbox{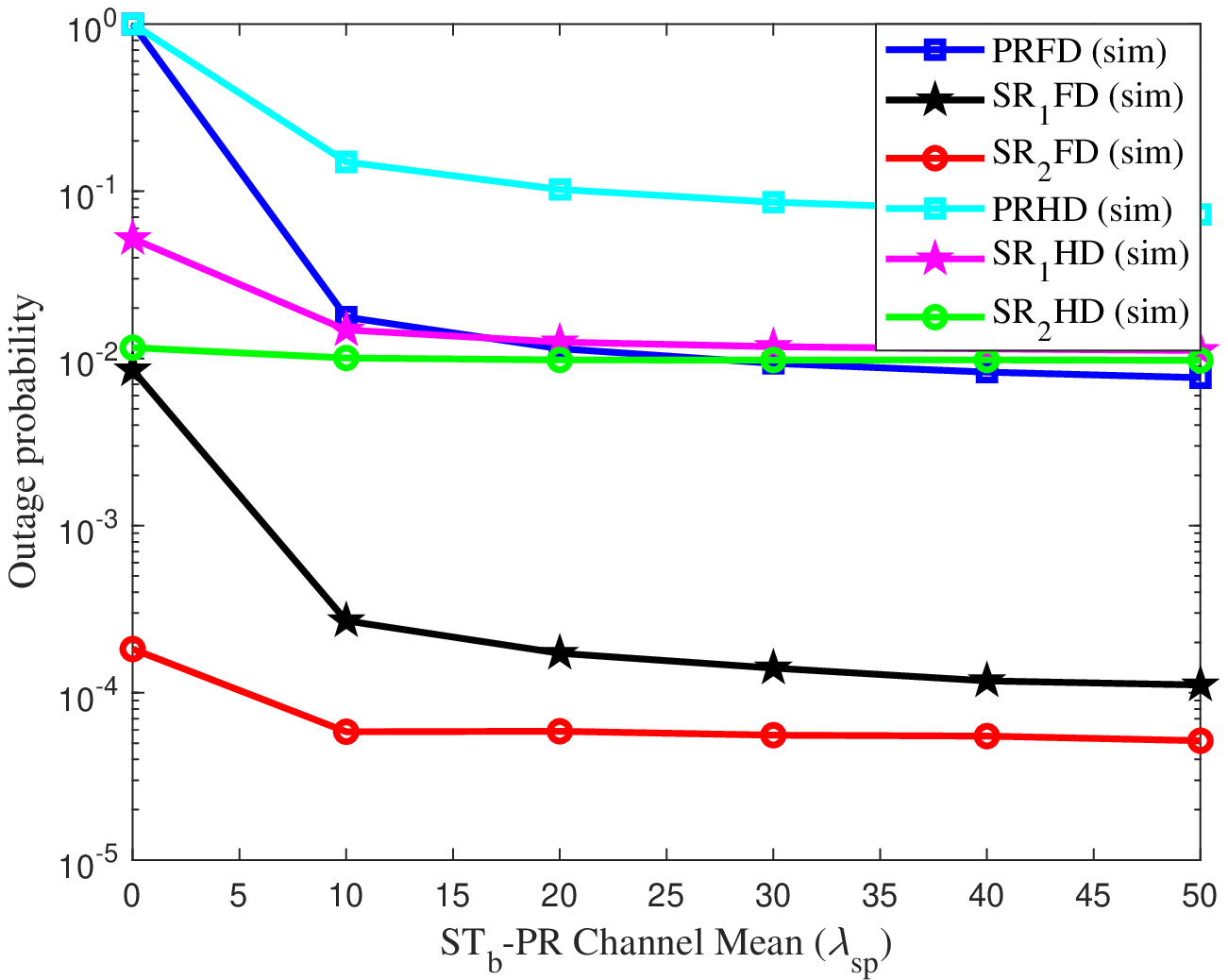} \caption{Impact of ST$_{\rm{b}}$-PR channel mean on the outage probability of primary and secondary networks operating in FD and HD modes, where $\beta = 0.8$, $N = 5$, $K = 3$, $\bar \gamma = -5$ dB, $\lambda_{\rm{ps}} = 5$, and $\lambda_{\rm{sr}} = 50$.} \label{fig:STPR}
		\end{center}
	\end{minipage}
\end{figure}

Fig. \ref{fig:6} depicts the impact of PS ratio $\beta$ on the outage performance of the FD and HD primary and secondary networks. As $\beta$ increases, the outage probability of both networks significantly reduces and attains the minimal values. The reason is that the growth in $\beta$ allows the ST$_{\rm{b}}$ to harvest more energy and in turn enhances the ST$_{\rm{b}}$'s transmit power and the achievable data rate at the ST$_{\rm{b}}$ as shown in \eqref{eq:PSn} and \eqref{eq:RSbxo}, respectively, which improves the information reception at the PR and SRs. Nevertheless, as $\beta$ continues increasing and reaches 1, the outage probability of both networks goes up and attains 1 since more power is given for energy harvesting and less power is left for the ST$_{\rm{b}}$ to decode $x_{0}$. Hence, outage occurs since the ST$_{\rm{b}}$ is unable to decode the PT's signal.
	
Fig. \ref{fig:STPR} illustrates the outage performance of primary and secondary networks in the FD and HD modes when the mean of ST$_{\rm{b}}$-PR ($\lambda_{\rm{sp}}$) channel varies, respectively. It is clear that the outage performance of primary and secondary networks greatly improves for better channel quality (higher channel mean) of ST$_{\rm{b}}$-PR link for both FD and HD modes. Note that similar results of the outage performance of primary and secondary networks in FD and HD modes are obtained with the increase in the mean of PT-ST$_{\rm{b}}$ ($\lambda_{\rm{ps}}$) or ST$_{\rm{b}}$-SRs ($\lambda_{\rm{sr}}$) channels as a result of the improvement in channel quality.
	
\subsection{System Throughput}
	
\begin{figure}[!h]
\begin{center}
\epsfxsize=9cm \epsfbox{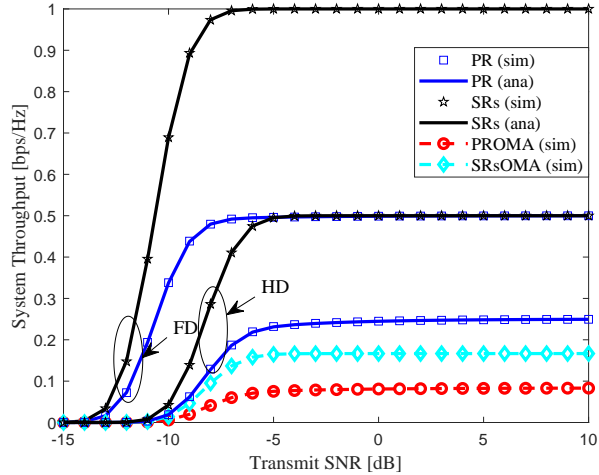} \caption{System throughput of primary and secondary networks operating in FD and HD modes, where $\beta = 0.8$, $N = 5$, and $K = 3$.} \label{fig:15}
\end{center}
\end{figure}  
	
The system throughput of primary and secondary networks in FD and HD modes is illustrated in Fig. \ref{fig:15}. It is clear that our proposed scheme obtains much higher system throughput than HD and the conventional OMA-TDMA schemes due to its lower outage probability. Besides, for all schemes, as the transmit SNR increases, the system throughput of primary and secondary networks goes up and reaches the system throughput floor, as presented in Remark 2.   

\subsection{Sum Rate}

\begin{figure}[!h]
	\begin{minipage}{0.46\linewidth}
		\begin{center}
			\epsfxsize=8cm \epsfbox{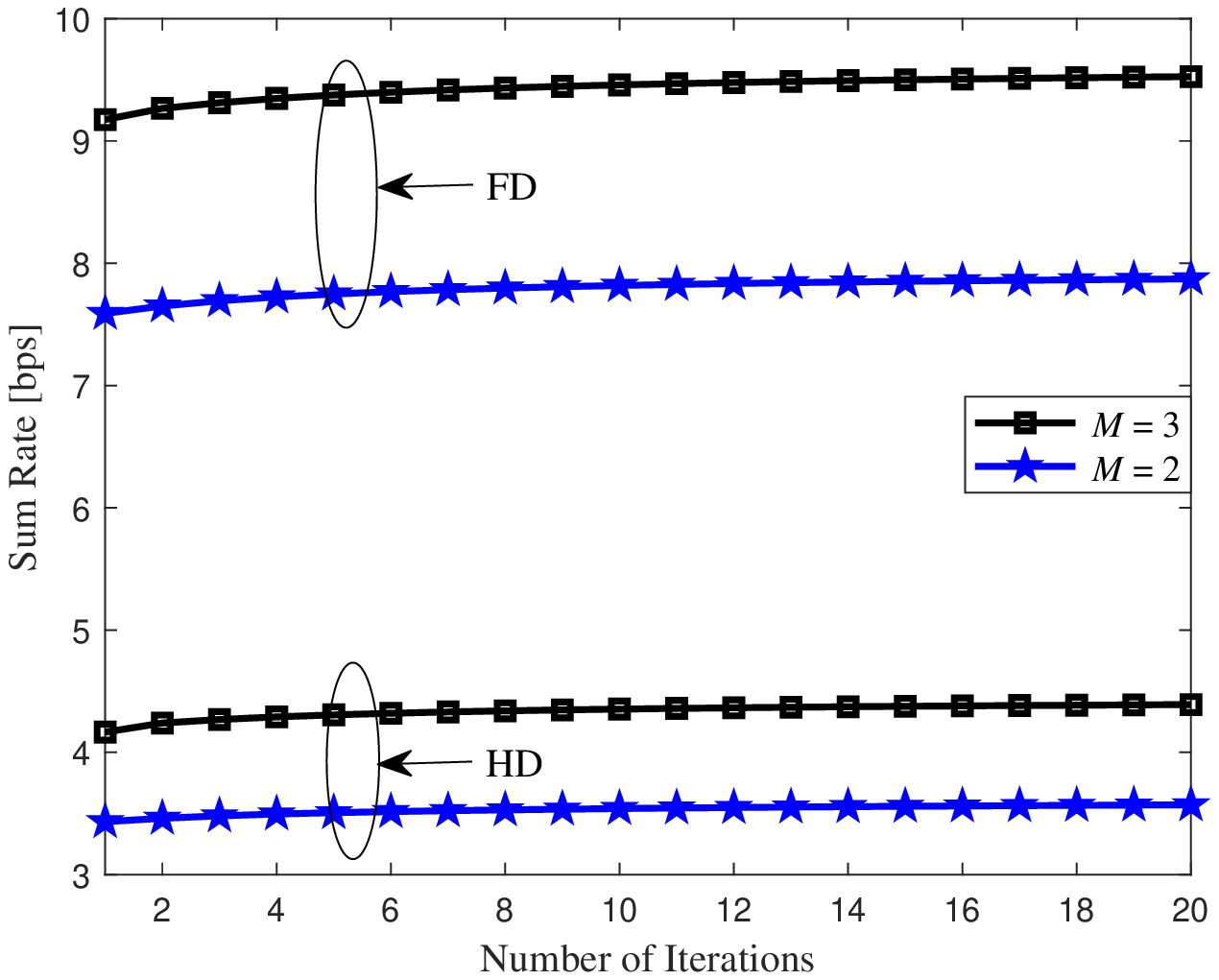} \caption{Convergence of the proposed sum rate maximization algorithm with different number of SRs $M$ in FD and HD modes, where $\beta = 0.8$, $N = 5$, $K = 3$, and $\bar \gamma = -5$ dB.} \label{fig:16}
		\end{center}
	\end{minipage}
	\hfill
	\begin{minipage}{0.46\linewidth}
		\begin{center}
			\epsfxsize=8cm \epsfbox{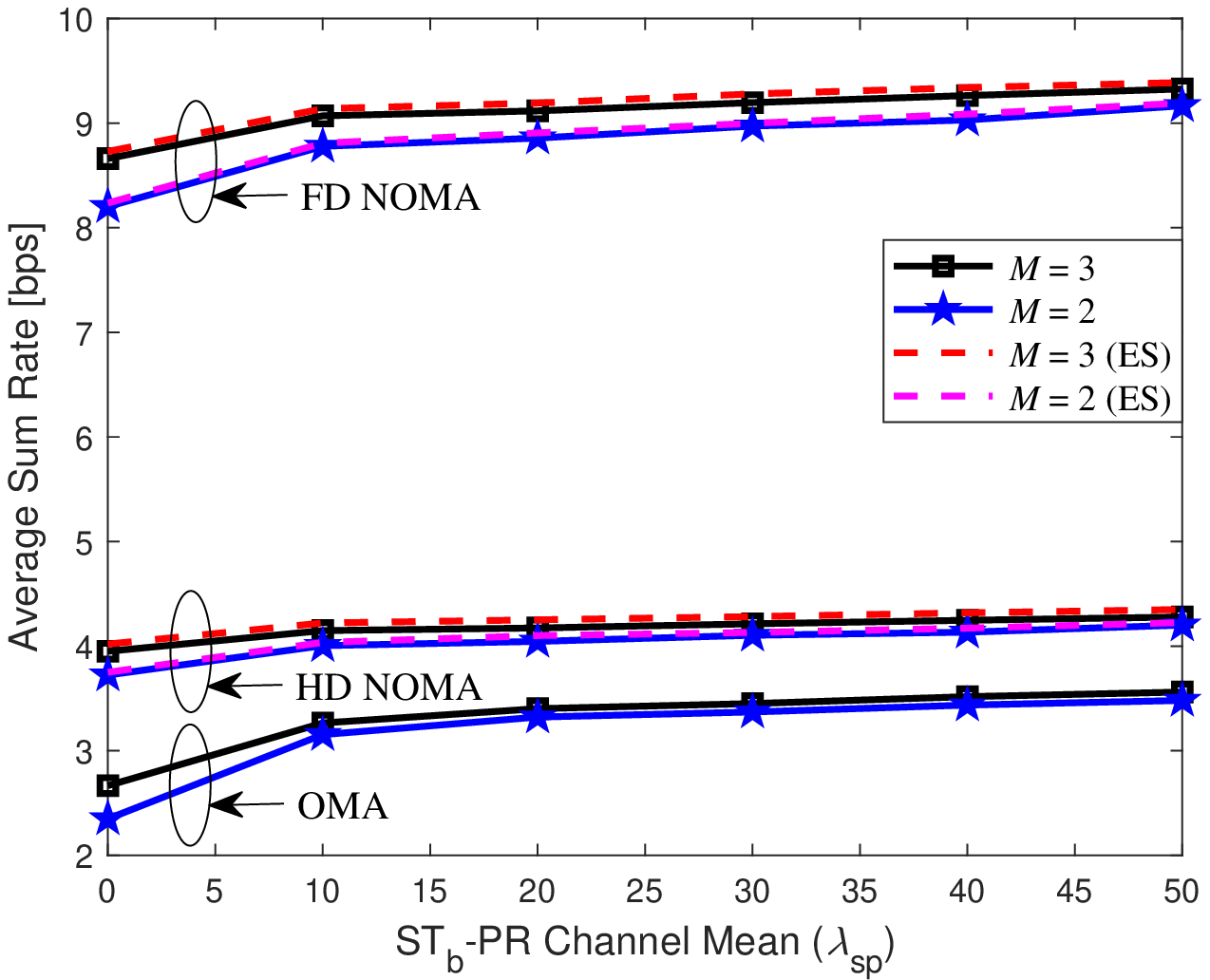} \caption{Effect of ST$_{\rm{b}}$-PR channel mean on the average sum rate of the proposed network in FD and HD modes, where $\beta = 0.8$, $N = 5$, $K = 3$, $\bar \gamma = -5$ dB, $\lambda_{\rm{ps}} = 5$, and $\lambda_{\rm{sr}} = 50$.} \label{fig:ASRSTPR}
		\end{center}
	\end{minipage}
\end{figure}

Fig. \ref{fig:16} illustrates the convergence behavior of our proposed sum rate maximization algorithm for the network when the ST$_{\rm{b}}$ operates in FD and HD modes, respectively. It can be observed that the proposed sum rate maximization algorithm only requires a small number of iterations to converge in both modes. Besides, increasing the number of SRs $M$ can significantly augment the sum rate of our proposed network in both modes.

Fig. \ref{fig:ASRSTPR} demonstrates the impact of ST$_{\rm{b}}$-PR ($\lambda_{\rm{sp}}$) channel mean on the average sum rate of the network when the FD and HD operation modes are used, respectively. It is obvious that the average sum rate greatly enhances with better channel quality (higher channel mean) of the ST$_{\rm{b}}$-PR link for both modes. Note that similar results of the average sum rate in FD and HD modes are obtained with the increase in the mean of PT-ST$_{\rm{b}}$ ($\lambda_{\rm{ps}}$) or ST$_{\rm{b}}$-SRs ($\lambda_{\rm{sr}}$) channels as a result of the enhancement in the channel quality. Furthermore, the average sum rate achieved by the proposed system using NOMA is higher than that in OMA-TDMA scheme. Moreover, we also compare the performance of Algorithm \ref{alg_IterativeAlgorithm} with the exhaustive search (ES), which is the method to find the globally optimal solution by searching all possible combinations of power allocation values. We can see that the average sum rate achieved by Algorithm \ref{alg_IterativeAlgorithm} is very close to the globally optimal solution found by the ES method. Therefore, our proposed iterative algorithm is able to achieve near-optimal performance with low complexity.

\begin{figure}[!h]
	\begin{center}
		\epsfxsize=9cm \epsfbox{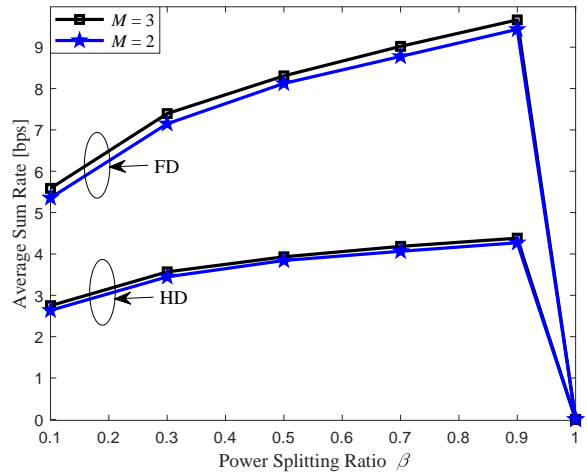} \caption{Impact of PS coefficient $\beta$ on the average sum rate of the proposed network in FD and HD modes, where $N = 5$, $K = 3$, and $\bar \gamma = -5$ dB.} \label{fig:17}
	\end{center}
\end{figure}  

Next, the effect of the PS coefficient $\beta$ on the average sum rate in the FD and HD cases is examined in Fig. \ref{fig:17}. It is seen that as $\beta$ increases, the average sum rate of the network operating in both FD and HD modes considerably grows for different number of SRs. The achieved results are reasonable since the rise in $\beta$ enables the ST$_{\rm{b}}$ to scavenge more energy and in turn improves the transmit power of the ST$_{\rm{b}}$, and hence, enhances the achievable data rates of the PR and SRs. However, when $\beta$ reaches 1, all RF power is harvested and no power is available for the PT to the ST$_{\rm{b}}$ information processing. Hence, the sum transmission rate of primary and secondary networks equals to 0 as the ST$_{\rm{b}}$ is unable to decode the PT's signal. 
	
\section{Conclusion}\label{sec:con}
In this paper, a novel NOMA assisted cooperative spectrum-sharing network has been proposed to encourage the collaboration between primary and secondary networks. In particular, the PT needs the help from the ST$_{\rm{b}}$ to relay its information and as an incentive, the ST$_{\rm{b}}$ has access to primary network's spectrum and uses it to transmit primary and secondary messages concurrently by using NOMA. The outage probability of primary and secondary networks has been derived in tight closed-form approximate expressions, and used to evaluate the system performance. Besides, the power allocation coefficients at ST$_{\rm{b}}$ have been optimized to maximize the sum transmission rate of the primary and secondary networks. Numerical results have showed the superior performance of the FD-NOMA system compared to the HD-NOMA and the conventional OMA-TDMA systems. Furthermore, it has been seen that several system design parameters, i.e., number of ST transmit antennas $N$, number of STs $K$, number of SRs $M$, the power allocation coefficients $\alpha$, power splitting ratio $\beta$, which have high impact on the system performance, should be carefully chosen in order to optimize the performance of the FD-SWIPT-NOMA cooperative spectrum-sharing network when applied in practice.
	
\appendices
\section{Proof of Lemma 1}\label{lemma1}

By using that $||\textbf{h}_{{\rm{p}}{\rm{s}}_{k}}||^2$ has the Gamma distribution and the PT-ST$_{k}$ links are independent, the CDF of $||\textbf{h}_{{\rm{p}}\rm{s_b}}||^2 = \mathop {\max }\limits_{k=1,2,\ldots,K} ||\textbf{h}_{{\rm{p}}{\rm{s}}_{k}}||^2$ can be derived as in \eqref{eq:Fhpsb}, where $C_{j}= 1$ for $j = 0$, $C_{j}= l$ for $j = 1$, and $C_{j}= \frac{1}{j} \sum\limits_{{p \rm{ = 1}}}^{q} \frac{pl - j + p}{p!} C_{j-p}$ for $2 \leq j \leq l(N - 1)$. $q = \min(j,N - 1)$. By differentiating \eqref{eq:Fhpsb}, the PDF of $||\textbf{h}_{{\rm{p}}\rm{s_b}}||^2$ can be obtained as in \eqref{eq:Phpsb}.
	
\section{Proof of Theorem 1}\label{theorem1}
	
By denoting the first term $\Pr \left( R_{{\Sbb},x_{0}} < \bar{R}_0 \right)$ and the second term $\Pr\left( R_{{\Sbb},x_{0}} \geq \bar{R}_0, R_{\PR} < \bar{R}_0 \right)$ in \eqref{eq:PP} as $\Phi_{1}$ and $\Phi_{2}$, respectively, and substituting \eqref{eq:RSbxo} and \eqref{eq:Fhpsb} into $\Phi_{1}$, $\Phi_{1}$ can be rewritten as:
\begin{align}\label{eq:Phi1}
\Phi_{1} &= \Pr \left[ \frac{(1-\beta) \bar \gamma ||\textbf{h}_{{{\rm{p}}\rm{s}}_{\rm{b}}}||^{2}}{(1-\beta) \rho \bar \gamma I_{\rm{SI}} ||\textbf{h}_{{{\rm{p}}\rm{s}}_{\rm{b}}}||^{2}+1} < \gamma_0 \right] \nonumber \\
&= \Pr \left[ ||\textbf{h}_{{{\rm{p}}\rm{s}}_{\rm{b}}}||^{2} < \frac{\gamma_0}{\bar \gamma \left( 1 - \beta \right) \left( 1 - \gamma_0 \rho I_{\rm{SI}} \right)} \right] \nonumber \\
&= \Pr \left( ||\textbf{h}_{{{\rm{p}}\rm{s}}_{\rm{b}}}||^{2} < \frac{\mu}{\bar \gamma} \right) \nonumber \\
&= \sum\limits_{{l\rm{ = 0}}}^{K} {\left( \begin{array}{l}
{K}\\
{l}
\end{array} \right)} (-1)^l \exp \left( -\frac{l \mu}{\lambda _{\rm{ps}} \bar \gamma} \right) \sum\limits_{{j\rm{ = 0}}}^{l(N - 1)} C_{j} \left(\frac{\mu}{{\lambda _{\rm{ps}}} \bar \gamma}\right)^{j},
\end{align}  
where $\gamma_0 < 1/(\rho I_{\rm{SI}})$; otherwise $\Phi_{1} = 1$, which leads to $\Phi_{2}$ in \eqref{eq:PP} equalling to 0 and hence $P_{\rm{p}} = 1$.
	
Next, we derive the term $\Phi_{2}$ in \eqref{eq:PP}. Substituting \eqref{eq:RSbxo} and \eqref{eq:R_P} into \eqref{eq:PP}, $\Phi_{2}$ can be rewritten as:
\begin{align}\label{eq:Phi2}
\Phi_2 &= \Pr \left[ ||\textbf{h}_{{{\rm{p}}\rm{s}}_{\rm{b}}}||^{2} \geq \frac{\gamma_0}{\bar \gamma \left( 1 - \beta \right) \left( 1 - \gamma_0 \rho I_{\rm{SI}} \right)}, ||\textbf{h}_{{{\rm{p}}\rm{s}}_{\rm{b}}}||^{2} \gSbP < \frac{\gamma_0}{\bar \gamma \rho \bigg( \alpha_0 - \sum\limits_{{i} = 1}^{M} \alpha_i \gamma_0 \bigg)}\right] \nonumber \\
& = \Pr \left( ||\textbf{h}_{{{\rm{p}}\rm{s}}_{\rm{b}}}||^{2} \geq \frac{\mu}{\bar \gamma} , ||\textbf{h}_{{{\rm{p}}\rm{s}}_{\rm{b}}}||^{2} \gSbP < \frac{\Theta_0}{\bar \gamma}\right),
\end{align}
where $\gamma_{0} < \min \left[ 1/(\rho I_{\rm{SI}}), \alpha_0/\left(\sum\limits_{{i} = 1}^{M} \alpha_i\right) \right]$. If $\gamma_0 > \alpha_0/\left(\sum\limits_{{i} = 1}^{M} \alpha_i\right)$, $P_{\rm{p}} = 1$.
	
With the use of the binomial theorem [28, Eq. (1.111)] and order statistics \cite{5}, the CDF of $\gSbP$ is shown as follows:
\begin{align}\label{eq:FgSbP}
F_{\gSbP}(x) &= \iota_{q} \sum\limits_{{c} = 0}^{Q-q} {\left( \begin{array}{cc}
{Q-q}\\
{c}
\end{array} \right)} \frac{(-1)^{c}}{q+c} \left[F_{|\bar h_{\rm{s}_{\rm{b}} {\rm{p}}}|^2}(x)\right]^{q+c} \nonumber \\
&= \iota_{q} \sum\limits_{{c} = 0}^{Q-q} {\left( \begin{array}{cc}
{Q-q}\\
{c}
\end{array} \right)} \frac{(-1)^{c}}{q+c} \sum\limits_{{n} = 0}^{q+c} {\left( \begin{array}{cc}
{q+c}\\
{n}
\end{array} \right)} (-1)^n \exp\left(-\frac{n x}{\lambda _{\rm{sp}}}\right),
\end{align}
where $|\bar h_{\rm{s}_{\rm{b}} {\rm{p}}}|^2$ is the unsorted channel gain of ST$_{\rm{b}}$-PR link and $F_{|\bar h_{\rm{s}_{\rm{b}} {\rm{p}}}|^2}(x) = 1-\exp\left(-\frac{x}{\lambda _{\rm{sp}}}\right)$. Besides, $q=1$ for PR and $q=m+1$ for SR$_{m}$.
	
In the case that $\gamma_{0} < \min \left[ 1/(\rho I_{\rm{SI}}), \alpha_0/\left(\sum\limits_{{i} = 1}^{M} \alpha_i\right) \right]$ and using \eqref{eq:Phpsb} and \eqref{eq:FgSbP}, \eqref{eq:Phi2} can be further derived as:
\begin{align}\label{eq:Phi21}
\Phi_2 & = \int\limits_{\mu/\bar\gamma}^{\infty} \Pr\left( \gSbP < \frac{\Theta_0}{\bar \gamma x} \right) f_{||\textbf{h}_{{\rm{p}}{\rm{s}}_{\rm{b}}}||^2}(x) dx \nonumber \\
& = \int\limits_{\mu/\bar\gamma}^{\infty} \iota_{q} \sum\limits_{{c} = 0}^{Q-q} \sum\limits_{{n} = 0}^{q+c} {\left( \begin{array}{cc}
{Q-q}\\
{c}
\end{array} \right)} {\left( \begin{array}{cc}
{q+c}\\
{n}
\end{array} \right)} \frac{(-1)^{c+n}}{q+c} \exp\left(-\frac{n \Theta_0}{\lambda _{\rm{sp}} \bar \gamma x}\right) \sum\limits_{{l\rm{ = 1}}}^{K} {\left( \begin{array}{l}
{K}\\
{l}
\end{array} \right)} {\left( { - 1} \right)^{{l}}} \exp \left( - \frac{l x}{\lambda _{\rm{ps}}} \right) \nonumber \\
& \times \sum\limits_{{j\rm{ = 0}}}^{l(N - 1)} \frac{C_{j}}{{\lambda _{\rm{ps}}}^j} \left( j x^{j - 1} - \frac{l}{\lambda _{\rm{ps}}} x^{j} \right) dx \nonumber \\
&= \int\limits_{\mu/\bar\gamma}^{\infty} \iota_{q} \widetilde \sum \Big(\sim\Big) \frac{(-1)^{c+n+l}}{q+c} \exp\left(-\frac{n \Theta_0}{\lambda _{\rm{sp}} \bar \gamma x}\right) \exp \left( - \frac{l x}{\lambda _{\rm{ps}}} \right) \frac{C_{j}}{{\lambda _{\rm{ps}}}^j} \left( j x^{j - 1} - \frac{l}{\lambda _{\rm{ps}}} x^{j} \right) dx,
\end{align}
where $\widetilde \sum = \sum\limits_{{c} = 0}^{Q-q} \sum\limits_{{n} = 0}^{q+c} \sum\limits_{{l\rm{ = 1}}}^{K} \sum\limits_{{j\rm{ = 0}}}^{l(N - 1)}$ and $\Big(\sim\Big) = {\left( \begin{array}{cc}
{Q-q}\\
{c}
\end{array} \right)} {\left( \begin{array}{cc}
{q+c}\\
{n}
\end{array} \right)} {\left( \begin{array}{l}
{K}\\
{l}
\end{array} \right)}$.
	
Since the integral in \eqref{eq:Phi21} can not be further simplified, we use the following approximation $e^{-\alpha/x}\approx1-\alpha/x$ for large values of $|x|$ \cite{27}. Hence, \eqref{eq:Phi21} can be further given by
\begin{align}\label{eq:Phi22}
\Phi_2 & \approx \int\limits_{\mu/\bar\gamma}^{\infty} \iota_{q} \widetilde \sum \Big(\sim\Big) \frac{(-1)^{c+n+l}}{q+c} \exp \left( - \frac{l x}{\lambda _{\rm{ps}}} \right) \frac{C_{j}}{{\lambda _{\rm{ps}}}^j} \left( j x^{j - 1} - \frac{l}{\lambda _{\rm{ps}}} x^{j} \right) dx -  \int\limits_{\mu/\bar\gamma}^{\infty} \iota_{q} \widetilde \sum \Big(\sim\Big) \frac{(-1)^{c+n+l}}{q+c} \nonumber \\
&\times \frac{n \Theta_0}{\lambda _{\rm{sp}} \bar \gamma x} \exp \left( - \frac{l x}{\lambda _{\rm{ps}}} \right) \frac{C_{j}}{{\lambda _{\rm{ps}}}^j} \left( j x^{j - 1} - \frac{l}{\lambda _{\rm{ps}}} x^{j} \right) dx.
\end{align}
	
With the help of [28, Eq. (3.381.3)] and after some manipulation steps, \eqref{eq:Phi22} can be further solved as:
\begin{align}\label{eq:Phi23}
\Phi_2 & \approx \iota_{q} \widetilde \sum \left(\sim\right) \frac{(-1)^{c+n+l}}{q+c} \frac{C_{j}}{{\lambda _{\rm{ps}}}^j} \left[j\left( \frac{l}{\lambda _{\rm{ps}}} \right)^{-j} \Gamma\left( j, \frac{l \mu}{\lambda _{\rm{ps}} \bar\gamma} \right) - \left(\frac{l}{\lambda _{\rm{ps}}} \right)^{-j} \Gamma\left( j+1, \frac{l \mu}{\lambda _{\rm{ps}} \bar\gamma} \right) \right] \nonumber \\
&- \iota_{q} \widetilde \sum \left(\sim\right) \frac{(-1)^{c+n+l}}{q+c} \frac{n \Theta_0}{\lambda _{\rm{sp}} \bar \gamma} \frac{C_{j}}{{\lambda _{\rm{ps}}}^j} \left[ j\left( \frac{l}{\lambda _{\rm{ps}}} \right)^{-j+1} \Gamma\left(j-1, \frac{l \mu}{\lambda _{\rm{ps}} \bar\gamma}\right) - \left( \frac{l}{\lambda _{\rm{ps}}} \right)^{-j+1} \Gamma\left( j, \frac{l \mu}{\lambda _{\rm{ps}} \bar\gamma}\right)\right] \nonumber \\
& \approx \iota_{q} \widetilde \sum \left(\sim\right) \frac{(-1)^{c+n+l}}{q+c} \frac{C_{j}}{{l}^j} \bigg[ \left(j + \frac{n \Theta_0 l}{\lambda _{\rm{sp}} \lambda _{\rm{ps}} \bar \gamma} \right) \Gamma\left( j, \frac{l \mu}{\lambda _{\rm{ps}} \bar\gamma}\right) - \Gamma\left( j+1, \frac{l \mu}{\lambda _{\rm{ps}} \bar\gamma} \right) \nonumber \\
& - \frac{n \Theta_0 j l}{\lambda _{\rm{sp}} \lambda _{\rm{ps}} \bar \gamma} \Gamma\left(j-1, \frac{l \mu}{\lambda _{\rm{ps}} \bar\gamma}\right) \bigg].
\end{align}
	
Finally, substituting \eqref{eq:Phi1} and \eqref{eq:Phi23} into \eqref{eq:PP}, we obtain the desired result as in \eqref{eq:PP1}.  
	
\section{Proof of Theorem 2}\label{theorem2}
	
Denote the second term $\Pr\left( R_{{\Sbb},x_{0}} \geq R_0, \overline{P}_{{\rm{r}}_m} \right)$ in \eqref{eq:Prm} as $\Upsilon$. Since the first term in \eqref{eq:Prm} is already derived in \eqref{eq:Phi1}, we will derive the second term $\Upsilon$. $\overline{P}_{{\rm{r}}_m}$ in $\Upsilon$ can be written as:
\begin{equation}\label{eq:oPR}
\overline{P}_{{\rm{r}}_m} = 1 - \Pr \left( {\rm{E}}^{c}_{{\rm{r}}_m,0} \cap \ldots \cap {\rm{E}}^{c}_{{\rm{r}}_m,m} \right),
\end{equation}
where ${\rm{E}}^{c}_{{\rm{r}}_m,m'} = R_{{\rm{r}}_m,x_{m'}} \geq \bar{R}_{m'}$, $0 \leq m' \leq m$, denotes the event that SR$_{m}$ successfully decodes $x_{m'}$. From \eqref{eq:Rmxk} and \eqref{eq:Rmxm}, $\Pr \left( {\rm{E}}^{c}_{{\rm{r}}_m,m'} \right)$ is expressed as:
\begin{align}\label{eq:PR1}
&\Pr \left({\rm{E}}^{c}_{{\rm{r}}_m,m'}\right) \nonumber \\
&= \Pr \left[ \frac{\alpha_{m'} \rho \bar \gamma ||\textbf{h}_{{{\rm{p}}\rm{s}}_{\rm{b}}}||^{2} \gSbRm}{\sum\limits_{{i} = m'+1}^{M} \alpha_{i} \rho \bar \gamma ||\textbf{h}_{{{\rm{p}}\rm{s}}_{\rm{b}}}||^{2} \gSbRm + 1} \geq \gamma_{m'} \right] \nonumber \\
& = \Pr \left[ ||\textbf{h}_{{{\rm{p}}\rm{s}}_{\rm{b}}}||^{2} \gSbRm \geq \frac{\gamma_{m'}}{\bar \gamma \rho \left(\alpha_{m'} - \sum\limits_{{i = m'+1}}^{M} \alpha_i \gamma_{m'} \right)} \right] \nonumber \\
& =  \Pr \left( ||\textbf{h}_{{{\rm{p}}\rm{s}}_{\rm{b}}}||^{2} \gSbRm \geq \frac{\Theta_{m'}}{\bar \gamma} \right),
\end{align}  
when $\gamma_{m'} \leq \alpha_{m'} / \left(\sum\limits_{{i = m'+1}}^{M} \alpha_i\right)$; otherwise, SR$_{m}$ suffers from the outage. According to \eqref{eq:oPR} and \eqref{eq:PR1}, $\overline{P}_{{\rm{r}}_m}$ is revised as:   
\begin{align}\label{eq:PR2}
\overline{P}_{{\rm{r}}_m} & = 1 - \Pr \left[ ||\textbf{h}_{{{\rm{p}}\rm{s}}_{\rm{b}}}||^{2} \gSbRm \geq \frac{\max \left(\Theta_0,\Theta_1,\ldots,\Theta_m\right)}{\bar \gamma} \right] \nonumber \\
& = \Pr \left( ||\textbf{h}_{{{\rm{p}}\rm{s}}_{\rm{b}}}||^{2} \gSbRm \leq \frac{\Theta}{\bar \gamma} \right).
\end{align}
	
Next, substituting \eqref{eq:PR2} into $\Upsilon$ in \eqref{eq:Prm}, $\Upsilon$ is rewritten as:
\begin{equation}\label{eq:PR21}
\Upsilon = \Pr \left( ||\textbf{h}_{{{\rm{p}}\rm{s}}_{\rm{b}}}||^{2} \geq \frac{\mu}{\bar \gamma}, ||\textbf{h}_{{{\rm{p}}\rm{s}}_{\rm{b}}}||^{2} \gSbRm \leq \frac{\Theta}{\bar \gamma} \right).
\end{equation}
	
Following the same steps to derive $\Phi_{2}$ in Appendix \ref{theorem1}, $\Upsilon$ is lastly obtained as follows:
\begin{align}\label{eq:Upsilon}
\Upsilon &\approx \iota_{q} \widetilde \sum \Big(\sim\Big) \frac{(-1)^{c+n+l}}{q+c} \frac{C_{j}}{{l}^j} \bigg[ \left(j + \frac{n \Theta l}{\lambda _{\rm{sr}} \lambda _{\rm{ps}} \bar \gamma} \right) \Gamma\left( j, \frac{l \mu}{\lambda _{\rm{ps}} \bar\gamma}\right) - \Gamma\left( j+1, \frac{l \mu}{\lambda _{\rm{ps}} \bar\gamma} \right) \nonumber \\
& - \frac{n \Theta j l}{\lambda _{\rm{sr}} \lambda _{\rm{ps}} \bar \gamma} \Gamma\left(j-1, \frac{l \mu}{\lambda _{\rm{ps}} \bar\gamma}\right) \bigg].
\end{align} 
	
Finally, substituting \eqref{eq:Phi1} and \eqref{eq:Upsilon} into \eqref{eq:Prm}, the final result is obtained as in \eqref{eq:Prm2}.

\section{Proof of Proposition 1}\label{proposition 1}

Let us define $\mathcal{F}(\mathbf{t})\triangleq \prod_{i=0}^{M} t_i$. Note that $\mathcal{F}(\mathbf{t}) \geq \mathcal{F}^{\left(\tau\right)}(\mathbf{t}), \forall \mathbf{t}$, and $\mathcal{F}(\mathbf{t}^{\left(\tau\right)}) = \mathcal{F^{\left(\tau\right)}}(\mathbf{t}^{\left(\tau\right)})$. Further, $\mathcal{F^{\left(\tau\right)}}(\mathbf{t}^{\left(\tau + 1\right)}) > \mathcal{F^{\left(\tau\right)}}(\mathbf{t}^{\left(\tau\right)})$ whenever $(\mathbf{t}^{\left(\tau + 1\right)}) \neq (\mathbf{t}^{\left(\tau\right)})$ since the former and the latter are the optimal solution and a feasible point for \eqref{eq:max4}, respectively. Hence, $\mathcal{F}(\mathbf{t}^{\left(\tau + 1\right)}) \geq \mathcal{F}^{\left(\tau\right)}(\mathbf{t}^{\left(\tau + 1\right)}) > \mathcal{F}^{\left(\tau\right)}(\mathbf{t}^{\left(\tau\right)}) = \mathcal{F}(\mathbf{t}^{\left(\tau\right)})$, presenting that $(\mathbf{t}^{\left(\tau + 1\right)})$ is a better feasible point than $(\mathbf{t}^{\left(\tau\right)})$ for problem \eqref{eq:max2}. The sequence $(\mathbf{t}^{\left(\tau\right)})$ of improved
feasible points for \eqref{eq:max2} thus converges at least to a locally
optimal solution which satisfies the KKT conditions \cite{3op}. As a result, the objective value of \eqref{eq:max4} is monotonically increasing, i.e., $\prod_{i=0}^{M} t^{\left(\tau\right)}_i \geq \prod_{i=0}^{M} t^{\left(\tau - 1\right)}_i$.

\begingroup
\setstretch{0.98}
\bibliographystyle{IEEEtran}

\endgroup

\end{document}